\documentclass[journal]{IEEEtran}

\usepackage{amsmath}
\usepackage{amsthm}
\usepackage{amssymb}
\usepackage{amsfonts}
\usepackage{bm}
\usepackage{array}
\usepackage{silence}
\usepackage{caption}
\usepackage{subcaption}
\usepackage{textcomp}
\usepackage{stfloats}
\usepackage{xurl}
\usepackage{verbatim}
\usepackage{graphicx}
\usepackage{multirow}
\usepackage{booktabs}
\usepackage{cite}
\usepackage[table]{xcolor}
\usepackage{float} 
\usepackage{algorithm} 
\usepackage{algpseudocodex}
\usepackage{placeins}
\usepackage{enumitem}
\usepackage{balance}

\newcommand{\gray}[1]{\color{gray}#1\color{black}}

\renewcommand{\paragraph}[1]{\noindent {\bf #1}}
\newtheorem{theorem}{Theorem}
\newtheorem{definition}{Definition}
\newtheorem{corollary}{Corollary}

\usepackage{hyperref}
\hypersetup{
    colorlinks,
    linkcolor={green!60!black},
    citecolor={red!70!black},
    urlcolor={blue!70!black}
}

\begin{document}
\title{Making Wide Stripes Practical: Cascaded Parity LRCs for Efficient Repair and High Reliability}
\author{
  Fan Yu, Guodong Li, Si Wu, Weijun Fang, and Sihuang Hu
  \thanks{Fan Yu and Guodong Li are with the Key Laboratory of Cryptologic Technology and Information Security, Ministry of Education, and the School of Cyber Science and Technology, Shandong University, Qingdao, Shandong 266237, China (e-mail: fanyu@mail.sdu.edu.cn, guodongli@sdu.edu.cn). Si Wu is with the School of Computer Science and Technology, Shandong University, Qingdao, Shandong 266237, China (e-mail: siwu5938@gmail.com). Weijun Fang and Sihuang Hu are with the Key Laboratory of Cryptologic Technology and Information Security, Ministry of Education, and the School of Cyber Science and Technology, Shandong University, Qingdao, Shandong 266237, China, and also with the Quan Cheng Laboratory, Jinan 250103, China (e-mail: \{fwj, husihuang\}@sdu.edu.cn). }
}

\maketitle

\begin{abstract}
  Erasure coding with wide stripes is increasingly adopted to reduce storage overhead in large-scale storage systems. However, existing Locally Repairable Codes (LRCs) exhibit structural limitations in this setting: inflated local groups increase single-node repair cost, multi-node failures frequently trigger expensive global repair, and reliability degrades sharply. We identify a key root cause: local and global parity blocks are designed independently, preventing them from cooperating during repair. We present Cascaded Parity LRCs (CP-LRCs), a new family of wide stripe LRCs that embed structured dependency between parity blocks by decomposing a global parity block across all local parity blocks. This creates a cascaded parity group that preserves MDS-level fault tolerance while enabling low-bandwidth single-node and multi-node repairs. We provide a general coefficient-generation framework, develop repair algorithms exploiting cascading, and instantiate the design with CP-Azure and CP-Uniform. Evaluations on Alibaba Cloud show reductions in repair time of up to 41\% for single-node failures and 26\% for two-node failures.
\end{abstract}


% !TeX root = ../main.tex

\section{Introduction} \label{sec:intro}

Modern distributed storage systems increasingly rely on {\em erasure coding} to provide cost-effective fault tolerance at scale \cite{athlur25,kadekodi23,chen17,huang12,rashmi13,sathiamoorthy2013xor,muralidhar14}. Compared to replication, erasure coding significantly reduces storage overhead while improving system reliability (measured in mean-time-to-data-loss (MTTDL)) \cite{kim19,weatherspoon2002erasure}. It works by encoding a set of original data blocks into additional {\em parity blocks} to form a {\em stripe}, such that enough data and parity blocks within a stripe can reconstruct the original data. As the data volume continues to explode \cite{burgener22}, both industry \cite{backblaze-erasure-coding, vastdata-resilience, kadekodi23} and academia \cite{hu21} are exploring {\em wide stripes} (stripes with a large number of data blocks and a very small number of parity blocks) for extreme storage savings. For example, Vastdata \cite{vastdata-resilience} deploys wide stripes with 150 data blocks and four parity blocks, Google \cite{kadekodi23} deploys stripes with a width of 80, while the study \cite{hu21} deploys stripes with 128 data blocks and four parity blocks. Even a 1\% reduction in storage overhead translates into millions of dollars in operational savings \cite{huang12,kadekodi23}.

However, wide stripes with traditional Maximum Distance Separable (MDS) codes (e.g., Reed-Solomon Codes \cite{reed60}) incur prohibitive repair bandwidth \cite{hu21, shen25}. For example, repairing a single lost block in a $(128,4)$ Reed-Solomon stripe requires accessing 128 remaining nodes. This motivates the adoption of {\em Locally Repairable Codes (LRCs)} \cite{huang12, sathiamoorthy2013xor, kolosov20, kadekodi23}. An LRC partitions data blocks into {\em local repair groups} of small size to generate {\em local parity blocks}, reducing repair bandwidth for a single block failure; it also generates {\em global parity blocks} from all data blocks to preserve fault tolerance. For example, in a $(128, 2, 2)$ Azure LRC stripe \cite{huang12}, repairing a single lost block requires accessing only 64 nodes.

Recent wide stripe LRC approaches such as Azure's LRC \cite{hu21}, Azure's LRC+1 \cite{kolosov20}, and Google's Optimal/Uniform Cauchy LRCs \cite{kadekodi23} demonstrate their potential in large-scale systems. Yet, our analysis and prior studies reveal that existing approaches face three fundamental limitations as stripe width grows:
(1) Even though LRCs introduce locality, wide stripes inflate local group size, causing single-node repair (especially parity repair) to require significantly more bandwidth.
(2) Wide stripes substantially increase the probability of multi-node failures \cite{kadekodi23,yu23}. When multiple failures occur in the same local repair group, LRCs must fall back to expensive {\em global repair}, often accessing a large number of data blocks and global parity blocks, negating the benefits of LRCs.
(3) The combined effect of higher repair costs and increased failure probability sharply reduces MTTDL for wide stripes.

Across existing LRC designs, we observe a structural limitation: {\em the local and global parity blocks are designed independently}, even though they are linear combinations of the same data blocks. This independence prevents cooperation between parity blocks during repair, resulting in unnecessarily high repair bandwidth, particularly during parity failures. At the same time, {\em the number of parity blocks is small in wide stripe settings}, meaning that introducing structured dependency among them can significantly reduce parity repair cost. This leads to our central idea: {\em Couple the local and global parity blocks to create a repair-efficient cascaded structure}.

We propose {\em Cascaded Parity LRCs (CP-LRCs)}, a new class of wide stripe LRC constructions that create an explicit dependency between all local parity blocks and one global parity block. The core idea is to start from a standard $(k,r)$ MDS stripe and decompose the encoding coefficients of a global parity block, distributing them across the $p$ local parity blocks. This ensures the last global parity block equals the XOR of all $p$ local parity blocks, forming a cascaded parity group. At the same time, the MDS-level fault tolerance of the original stripe is maintained. This cascaded structure fundamentally changes how repairs occur in wide stripes:
(1) {\bf Parity repair becomes local}: a lost local or global parity block can be repaired using only a few nodes in the cascaded group.
(2) {\bf Multi-node repair becomes cheaper}: many multi-failure patterns that previously required global repair can now be handled through a sequence of local repairs.
(3) {\bf Repair latency and exposure time decrease}:  resulting in significantly higher MTTDL.

To make CP-LRCs practical, we design generic coefficient-generation rules and repair algorithms that preserve minimum distance, exploit cascading during repair, and support arbitrary $(k,r,p)$ parameters.  We further instantiate the framework with two system-relevant constructions: {\em CP-Azure} and {\em CP-Uniform}.

To summarize, this paper makes the following contributions:
\begin{itemize}
    \item We identify fundamental limitations of existing wide stripe LRCs in terms of single-node repair, multi-node failure handling, and reliability.
    \item We propose CP-LRCs, a new family of wide stripe LRCs that introduce structured dependency between local and global parity blocks.
    \item We design a general coefficient-construction framework that preserves the MDS-level fault tolerance while enabling cascading.
    \item We develop repair algorithms that exploit the cascaded structure to mitigate repair bandwidth for both single-node and multi-node failures.
    \item We implement CP-LRCs in a distributed storage prototype and evaluate performance in Alibaba Cloud \cite{alibaba25}. Results show that CP-LRCs reduce the single-node and two-node repair times of existing wide stripe LRCs by up to 41\% and 26\%, respectively.
\end{itemize}

We release the source code of our prototype with CP-LRCs at \url{https://github.com/Yf-holiday/Cascade-Parity-LRC.git}.
% !TEX root = ../main.tex

\section{Background} \label{sec:background}

\subsection{Erasure Coding} \label{sec:erasure-coding}

\paragraph{Basics.} {\em Erasure coding} is widely deployed in large-scale distributed storage systems (e.g., Ceph, HDFS) to provide fault tolerance with much lower storage overhead than replication \cite{chen17,rashmi13,sathiamoorthy2013xor,muralidhar14}. With erasure coding, data is typically organized into {\em stripes}, where $k$ data blocks are encoded into $n$ blocks (including the $k$ original data blocks and $n-k$ additional {\em parity blocks} that provide redundancy). The $n$ blocks of a stripe are distributed across $n$ distinct nodes for node-level fault tolerance. Thus, a {\em single-node failure} refers to the loss of a single block within a stripe, while a {\em multi-node failure} refers to the loss of multiple blocks in the same stripe. Maximum Distance Separable (MDS) codes are favored in distributed storage systems due to their optimal fault tolerance properties within the given storage overhead, i.e., a $(n, k)$ MDS code tolerates any $n-k$ node failures \cite{plank13}. However, MDS codes like Reed-Solomon Codes require accessing all $k$ data blocks to repair up to $n-k$ node failures, which incurs excessive bandwidth consumption \cite{li19fast}.

\paragraph{Locally Repairable Codes (LRCs).} To suppress the high repair cost in MDS codes, LRCs are adopted. An LRC with parameters $(k, r, p)$ encodes $k$ data blocks into $r$ {\em global parity blocks} similar to the MDS codes to provide fault tolerance. It partitions the data blocks (and possibly global parity blocks) into $p$ {\em local repair groups}, each containing up to $k$ blocks. By adding one {\em local parity block} to each local repair group, an LRC significantly reduces the repair bandwidth for single-node failures, i.e., it requires accessing only a small set of blocks for single-node repair. Due to their high repair efficiency, LRCs are widely employed in the commercial storage systems of major enterprises such as Google \cite{kadekodi23,athlur25}, Microsoft \cite{huang12,chen17}, and Facebook \cite{sathiamoorthy2013xor}.

\subsection{Wide Stripe LRCs}
\label{sec:wide-lrc}

\paragraph{Constructions.} As LRCs greatly enhance the single-node repair efficiency compared to MDS codes, they are utilized to construct wide stripes for further storage savings (Section~\ref{sec:intro}). Below, we introduce several representative wide stripe LRC constructions used in industry and academia, illustrated in Figure~\ref{fig:lrc_construction}.

\begin{figure*}
    \centering
    \begin{subfigure}{0.24\linewidth}
        \centering
        \includegraphics[width=0.9\linewidth]{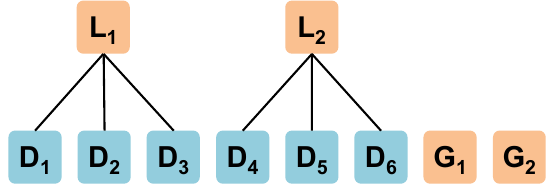}
        \caption{Azure LRC}
    \end{subfigure}
    \begin{subfigure}{0.24\linewidth}
        \centering
        \includegraphics[width=0.9\linewidth]{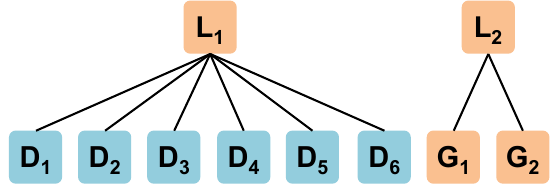}
        \caption{Azure LRC+1}
    \end{subfigure}
    \begin{subfigure}{0.24\linewidth}
        \centering
        \includegraphics[width=0.9\linewidth]{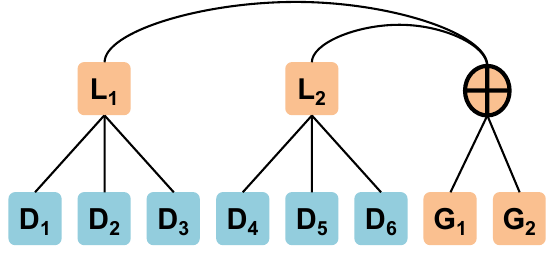}
        \caption{Optimal Cauchy LRC}
    \end{subfigure}
    \begin{subfigure}{0.24\linewidth}
        \centering
        \includegraphics[width=0.9\linewidth]{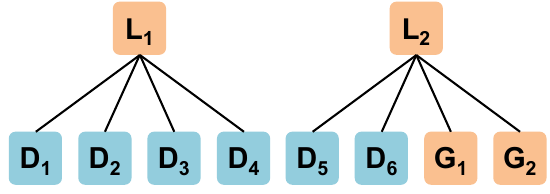}
        \caption{Uniform Cauchy LRC}
    \end{subfigure}
    \caption{ Illustration of LRC constructions with parameters $(k=6, r=2, p=2)$. $D_1, D_2, \ldots, D_6$ denote data blocks; $L_1$, $L_2$ are local parity blocks; and $G_1$, $G_2$ are global parity blocks.}
    \label{fig:lrc_construction}
\end{figure*}

{\em Azure LRC} \cite{huang12}: It generates the global parity blocks by encoding the data blocks using an $(r \times k)$ Vandermonde matrix \cite{plank09}. The construction assumes that $k$ is divisible by $p$, divides $k$ data blocks evenly into $p$ local repair groups, and within each group, generates a local parity block by XORing the $\tfrac{k}{p}$ data blocks. Furthermore, it tolerates any $r+1$ node failures across a stripe.

    {\em Azure LRC+1} \cite{kolosov20}: An $(k, r, p)$ Azure LRC+1 code is constructed from an $(k, r, p-1)$ Azure LRC code by adding a local parity block for all $r$ global parity blocks. As such, global parity blocks are protected by a local parity. It also tolerates any $r+1$ node failures as in Azure LRC.

    {\em Optimal Cauchy LRC} \cite{kadekodi23}: It constructs the global parity blocks based on a Cauchy matrix \cite{plank06} (in contrast to the Vandermonde matrix in Azure LRC). It also partitions $k$ data blocks evenly into $p$ local repair groups. To achieve optimal {\em minimum distance} ({\em distance} for short), it adds an XOR sum of all global parity blocks into each local repair group to generate a local parity block. It can tolerate any $r+1$ node failures.

    {\em Uniform Cauchy LRC} \cite{kadekodi23}: Similar to Optimal Cauchy LRC, it generates the global parity blocks from a Cauchy matrix. Then, it uniformly divides all data blocks and global parity blocks into local repair groups and generates the local parity blocks. It achieves small, uniform locality but tolerates any $r$ node failures.

\paragraph{Metrics.} To evaluate the efficiency of wide stripe LRCs in terms of repair and reliability, we introduce several metrics. Note that the same metrics are defined in \cite{kadekodi23}.

Average degraded read cost (ADRC) measures the average number of nodes that need to be accessed to repair all data blocks, defined as:
\[
    \text{ADRC} = \frac{1}{k}\sum_{i=1}^k \text{cost}(B_i),
\]
where $\text{cost}(B_i)$ denotes the number of nodes accessed for repairing block $B_i$.

Average single-node repair cost ($\text{ARC}_1$) extends $\text{ADRC}$ to all blocks in a stripe:
\[
    \text{ARC}_1 = \frac{1}{n}\sum_{i=1}^n \text{cost}(B_i).
\]

For MDS codes, $\text{cost}(B_i) = k$ for any block; however, for LRCs, the value depends on whether the block is a data block, local parity block, or global parity block.

Average two-node repair cost ($\text{ARC}_2$) further extends $\text{ARC}_1$ to two-node failure
cases:
\[
    \text{ARC}_2 = \frac{\sum_{i=1}^{n}\sum_{\substack{j=1 \\ j \neq i}}^{n} \text{cost}(B_i, B_j)}{\binom{n}{2}},
\]
where $\text{cost}(B_i, B_j)$ is the number of nodes accessed to repair both $B_i$ and $B_j$, which may involve {\em local repair} (i.e., performing repair by accessing the local parity blocks) or {\em global repair} (i.e., performing repair by accessing the global parity blocks) in LRCs.

\begin{figure}
    \centering
    \includegraphics[width=0.9\linewidth]{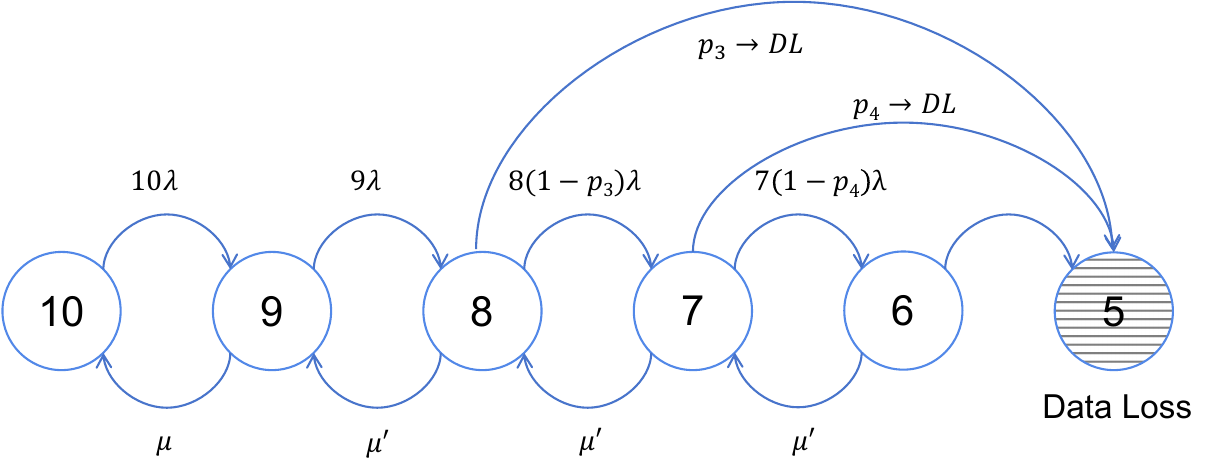}
    \caption{Markov chain for a $(6,2,2)$ LRC.}
    \label{fig:mttdl_markov}
\end{figure}

To evaluate reliability, we adopt the MTTDL metric as in prior studies (~\cite{hu21,sathiamoorthy2013xor,Cidon15,dimakis10,Silberstein14,greenan2010mean}).
The reliability process is modeled using a Markov chain, as shown in Figure~\ref{fig:mttdl_markov}, where each state presents the number of surviving nodes in a stripe.
For example, state 10 indicates that all 10 nodes are healthy, while state 5 indicates that five nodes in a stripe are failed, leading to data loss.

The transition from state $i$ to $i-1$ corresponds to a node failure with rate $i\lambda$, where $\lambda$ is the per-node failure rate.
When the number of failed nodes exceeds the number of global parity blocks ($r$), repair may fail with probability $p_i$, and the transition rate becomes $i(1-p_i)\lambda$.
Conversely, the transition from state $i-1$ to $i$ occurs via repair at rate $\mu_i$, which is primarily determined by the repair time for single-node failures and the failure detection time for multi-node failures.
Finally, MTTDL is computed from the steady-state probability distribution of this Markov chain.
\section{Challenges and Motivations} \label{sec:challenges_and_motivations}

\subsection{Challenges in Wide Stripe LRCs} \label{sec:challenges-wide-lrc}

\paragraph{Increased single-node repair cost.} While LRCs reduce the repair bandwidth for single-node failures by accessing a few blocks, wide stripe LRCs significantly increase the required bandwidth. For example, in a $(6, 2, 2)$ Azure LRC stripe, repairing a data block requires accessing three nodes, whereas in a $(24, 2, 2)$ Azure LRC stripe, it requires accessing 12 nodes. This is because wide stripe LRCs have larger local repair groups, leading to higher repair bandwidth even for single-node failures. This issue is exacerbated when repairing a global parity block (i.e., in Azure LRC), which requires accessing all $k$ data blocks, resulting in prohibitively high bandwidth cost.

\paragraph{More frequent multi-node repair and expensive cost.} In wide stripe LRCs, the likelihood of multi-node failures increases significantly as the stripe width (i.e., $n$) grows. A field study by Google \cite{kadekodi23} reports that its LRC deployment with a stripe width of 80 observes four times as many stripes with at least four failures as with a width of 50. Similarly, simulation results \cite{yu23} indicate that the fraction of stripes experiencing multi-node failures grows rapidly with increasing stripe width, reaching approximately 30\% when the width is 64.

However, multi-node failures are highly likely to trigger global repair, which typically involves accessing a large number of data blocks and some global parity blocks, thereby significantly increasing repair bandwidth \cite{kolosov20}. Specifically, since LRCs provide only one local parity block within each local repair group, they can tolerate and repair a single node failure locally. If more than one node fails within the same local repair group, the system must resort to global repair. For example, in a $(24, 2, 2)$ Azure LRC stripe, the failure of two data blocks in the same group triggers global repair, requiring access to 24 nodes.

\paragraph{Decreased reliability.} Wide stripe LRCs face a significant reliability challenge, as their MTTDL degrades significantly compared to narrow stripe configurations. The reasons are two-fold: (1) The probability of failures, particularly multi-node failures, increases with larger stripe width; and (2) the repair efficiency for both single-node and multi-node failures decreases as the number of blocks required for repair grows. For example, in a $(6, 2, 2)$ Azure LRC stripe, the MTTDL is $2.66 \times 10^{17}$ years, while in a wider $(24, 2, 2)$ Azure LRC stripe, the MTTDL drops significantly to $1.90 \times 10^{14}$. Such observation underscores the need to improve the reliability of wide stripe LRCs.

\subsection{Motivations for Cascaded Parity LRCs} \label{sec:motivation}

\paragraph{Key observations and idea.} In existing wide stripe LRC constructions, the local and global parity blocks are designed independently. For example, in Azure LRC, there is no dependency between the two, which results in high repair bandwidth: repairing a global parity block requires accessing $k$ nodes, while repairing a local parity block involves $\tfrac{k}{p}$ nodes. However, both the local and global parity blocks are essentially linear combinations of the original data blocks, differing only in their encoding coefficients. This reveals an opportunity to intentionally introduce a dependency between them to reduce the repair costs.

Moreover, in wide stripe settings, the number of local and global parity blocks is relatively small. Coupling them effectively can significantly lower the bandwidth required for parity repair. Based on these insights, we propose {\em Cascaded Parity LRCs (CP-LRCs)}, which create an explicit 'channel' between the local and global parity blocks by carefully designing their encoding coefficients. Specifically, CP-LRCs cascade local parity blocks with a global parity block, enabling parity repair within the cascaded group rather than the entire stripe. This design not only reduces parity repair bandwidth but also increases the likelihood of low-cost local repair during multi-node failures, thereby lowering the probability of invoking expensive global repair.

\paragraph{Benefits and motivations.} By cascading the local and global parity blocks, CP-LRCs bring three benefits, which also motivate our design.

    {\em Decreasing single-node repair cost.} CP-LRCs form a cascaded group that directly reduces the number of nodes accessed for parity repair. For example, a $(24, 2, 2)$ CP-Azure LRC (i.e., applying CP-LRC to Azure LRC) cascading $L_1$ and $L_2$ with $G_2$ reduces the number of nodes accessed to repair $L_1$/$L_2$/$G_2$ from 12/12/24 in the original Azure LRC to just two. Besides, cascading shortens the size of the original local repair groups. For example, in Figure~\ref{fig:cp-lrcs}(c), applying CP-LRC to Uniform Cauchy LRC ({\em CP-Uniform}) with the parameters of $(6,2,2)$ shortens the group size for $D_1,\ldots, D_{4}$ to 3 (from 4 in the original Uniform Cauchy LRC), further reducing the repair cost.

    {\em Improving the local repair ratio and reducing multi-node repair cost.} Cascading increases the number of local repair groups, which raises the probability that multi-node failures can still be resolved through local repair rather than triggering costly global repair. For example, in a $(24, 2, 2)$ CP-Azure LRC, if $D_1$ and $L_1$ fail simultaneously, the conventional approach (e.g., in Azure LRC) must invoke global repair, which accesses 24 nodes (i.e., $D_2, D_3,\ldots, D_{24}$ and $G_1$). In contrast, under CP-Azure, $L_1$ can be repaired within the cascaded group first, and then $D_1$ can be repaired within its local repair group. This two-step local repair involves only 13 nodes (i.e., $ D_2,D_3,\ldots,D_{12}$, $L_2$, and $G_2$).

    {\em Increasing reliability.} By lowering both single-node and multi-node repair costs, CP-LRCs accelerate the repair operations and reduce exposure time to data loss. This leads to improved system reliability (i.e., higher MTTDL) compared to conventional wide stripe LRCs.

\begin{table}[!t]
    \centering
    \caption{Comparison of Repair and Reliability of Different LRCs}
    \label{tab:lrc_performance}
    \resizebox{\linewidth}{!}{
        \begin{tabular}{|c|c|c|c|c|c|}
            \hline
            \textbf{Parameters} & \textbf{Scheme}    & \textbf{ADRC} & \textbf{ARC$_1$} & \textbf{ARC$_2$} & \textbf{MTTDL} \\ \hline
            \multirow{4}{*}{(6,2,2)}
                                & Azure LRC          & 3.00          & 3.60             & 6.00             & 2.66e+17       \\
                                & Azure LRC+1        & 6.00          & 4.80             & 6.22             & 1.99e+17       \\
                                & Optimal Cauchy LRC & 5.00          & 5.00             & 6.27             & 1.91e+17       \\
                                & Uniform Cauchy LRC & 4.00          & 4.00             & 6.22             & 2.39e+17       \\
                                & CP-Azure           & 3.00          & 3.00             & 5.80             & 3.19e+17       \\
                                & CP-Uniform         & 3.50          & 3.10             & 6.00             & 3.09e+17       \\ \hline
            \multirow{4}{*}{(24,2,2)}
                                & Azure LRC          & 12.00         & 12.86            & 24.00            & 1.90e+14       \\
                                & Azure LRC+1        & 24.00         & 21.64            & 24.07            & 1.13e+14       \\
                                & Optimal Cauchy LRC & 13.00         & 13.00            & 25.17            & 1.89e+14       \\
                                & Uniform Cauchy LRC & 14.00         & 13.00            & 24.07            & 1.89e+14       \\
                                & CP-Azure           & 12.00         & 11.36            & 21.82            & 2.16e+14       \\
                                & CP-Uniform         & 12.50         & 11.39            & 22.03            & 2.32e+14       \\ \hline
        \end{tabular}
    }
\end{table}

\paragraph{Direct demonstration.} Table~\ref{tab:lrc_performance} presents a direct performance comparison between existing LRC schemes and our CP-LRCs. CP-Azure and CP-Uniform consistently achieve the smallest and second smallest $\text{ARC}_1$ and $\text{ARC}_2$ values among all schemes, indicating superior single- and multi-node repair efficiency. They also exhibit low ADRC values, reflecting reduced degraded read bandwidth. Furthermore, both CP-Azure and CP-Uniform significantly improve MTTDL, demonstrating greater reliability than existing wide stripe LRC constructions.

\section{Code Design} \label{sec:design}

\begin{figure}
  \centering
  \begin{subfigure}{0.485\linewidth}
    \centering
    \includegraphics[width=0.95\linewidth]{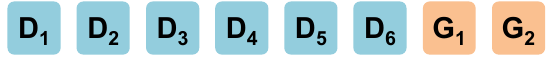}
    \caption{A ($4, 2$) MDS stripe}
  \end{subfigure}
  \\[.5em]
  \begin{subfigure}{0.49\linewidth}
    \centering
    \includegraphics[width=0.95\linewidth]{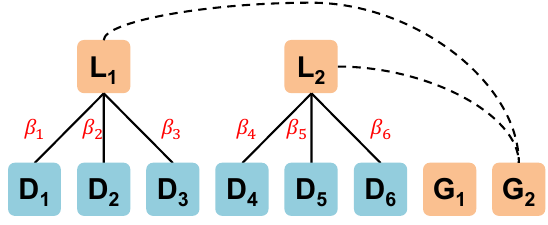}
    \caption{A ($6, 2, 2$) CP-Azure LRC}
  \end{subfigure}
  \begin{subfigure}{0.49\linewidth}
    \centering
    \includegraphics[width=0.95\linewidth]{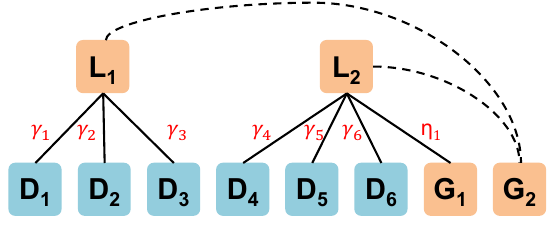}
    \caption{A ($6, 2, 2$) CP-Uniform LRC}
  \end{subfigure}
  \caption{
    Illustration of CP-LRCs applied to Azure LRC and Uniform LRC with parameters $(6,2,2)$. $\beta_1, \beta_2,\ldots ,\beta_{6}$ and $\gamma_1 ,\gamma_2,\ldots ,\gamma_{6},\eta_1$ are the encoding coefficients for local parities. Note that $L_1$, $L_2$, and $G_2$ form a cascaded group, i.e., both CP-Azure and CP-Uniform satisfy that $G_2 = L_1 + L_2$.
  }
  \label{fig:cp-lrcs}
\end{figure}

The previous section shows that {\em Cascaded Parity LRCs (CP-LRCs)} substantially improve the repair efficiency and reliability by coupling all local parity blocks with a global parity block. Although the high-level idea of cascading appears simple, realizing CP-LRCs in a robust and practical form is far from straightforward. A careful design is required to ensure that the cascaded structure delivers the intended repair benefits while preserving the original fault tolerance of LRCs and adapting to general coding parameters.

  {\bf Remaining challenges:} First, the cascaded structure relies on a precise selection of encoding coefficients. Unlike conventional LRCs, where the local and global parity coefficients are designed independently, CP-LRCs introduce dependencies that must be crafted to ensure the resulting code enables cascading while maintaining its minimum distance. Second, achieving efficient single-node and multi-node repairs requires coordinated interactions between local repair groups and the cascaded group, ensuring that failures are handled locally whenever possible. This requires repair algorithms that fully exploit the cascaded structure without triggering unnecessary global repair. Finally, to make CP-LRCs broadly applicable, the design must generalize to arbitrary $(k, r, p)$ parameters, which introduces additional constraints on coefficient generation.

In this section, we address these challenges by presenting a complete design of CP-LRCs. We describe how to construct encoding coefficients that enable reliable cascading, how the single-node and multi-node repair algorithms exploit this structure, and how the coefficient design extends to general parameters.

\subsection{Overview of CP-LRC Framework} \label{subsec:desgin_overview}

A CP-LRC code with parameters $(k,r,p)$ starts from a standard $(k,r)$ MDS-coded stripe with $k$ data blocks and $r$ global parity blocks (e.g., a Cauchy Reed-Solomon stripe \cite{plank06}). It then augments the stripe with $p$ local parity blocks by decomposing the encoding coefficients of the last global parity block and distributing them across the $p$ local parity blocks. With this construction, the last global parity block equals the XOR sum of all $p$ local parity blocks, forming a cascaded structure while preserving the MDS-level fault tolerance of the underlying stripe.  This cascaded design enables both single-node and multi-node repairs to be executed within local repair groups and the cascaded group, improving performance. We next show the base construction of CP-LRCs (Section~\ref{sec:base}), and then instantiate our framework with two concrete designs: CP-Azure (Section~\ref{sec:cp-azure}) and CP-Uniform (Section~\ref{sec:cp-uniform}).

\subsection{Starting Point: The Base MDS Stripe} \label{sec:base}

We build the base stripe using a systematic $(k,r)$ MDS code over $\mathrm{GF}(2^w)$. The stripe contains $k$ data blocks $D_1,\ldots,D_k$ and $r$ global parity blocks $G_1,\ldots,G_r$. Each global parity block is a linear combination of all $k$ data blocks with nonzero coefficients. Mathematically, for global parity block $G_j$ (where $1\le j\le r$),
\begin{equation}\label{eq:MDS-encoding}
  G_j = \alpha_{1,j}D_1 + \alpha_{2,j}D_2 + \cdots + \alpha_{k,j}D_k,
\end{equation}
where $\alpha_{i,j}$ are nonzero coefficients in the Galois Field $\mathrm{GF}(2^w)$.

\paragraph{MDS property.} A $(k,r)$ MDS-coded stripe tolerates any $r$ simultaneous block failures; or equivalently, any $k$ blocks (data or parity) suffice to reconstruct the entire stripe ( \cite{rawat14,blaum2013partial}). Repairing a single failed block requires $k$ surviving blocks.

In our CP-LRC design, the base stripe can be instantiated with any standard MDS code (e.g., Vandermonde matrix-based Reed-Solomon Codes~\cite{reed60} and Cauchy matrix-based Reed-Solomon Codes~\cite{plank06}).

\paragraph{A $(6,2)$ MDS-coded stripe.}
Figure~\ref{fig:cp-lrcs}(a) shows a concrete MDS-coded stripe with parameters $(6,2)$, which comprises six data blocks $D_1,D_2,\ldots,D_6$ and two global parity blocks $G_1,G_2$. The two global parity blocks are linear combinations of the six data blocks, i.e.,
\begin{equation}\label{eq:G_1}
  G_1 = \alpha_{1,1}D_1 + \alpha_{2,1}D_2 +\cdots+\alpha_{6,1}D_6,
\end{equation}
\begin{equation}
  G_2 = \alpha_{1,2}D_1 + \alpha_{2,2}D_2 +\cdots+\alpha_{6,2}D_6.
\end{equation}\label{eq:G_2}
In this $(6,2)$ MDS stripe, any $6$ blocks (data or parity) can reconstruct the whole stripe.

\subsection{CP-Azure LRC} \label{sec:cp-azure}

To build the local repair groups of the $(k, r, p)$ CP-Azure LRC, we first partition the $k$ data blocks of the base $(k, r)$ MDS stripe into $p$ groups evenly. Then, we build $p$ local parity blocks (denoted $L_1, L_2, \ldots, L_p$) by linearly combining the data blocks in each group.

The building of local parity blocks in CP-Azure LRC satisfies the rule that the last global parity block $G_r$ equals the field-sum (XOR) of all local parity blocks, i.e,
\begin{equation}\label{eq:implied-repair}
  L_1 + L_2 + \cdots + L_p = G_r.
\end{equation}
Thus, we generate local parity blocks in CP-Azure LRC using the coefficients in which the last global parity block $G_r$ is formed from data blocks in the base MDS stripe. Recall that the last global parity block $G_r$ in the base MDS stripe is computed as
\begin{equation}\label{eq:xxx-1}
  G_r = \beta_{1}D_1 + \beta_{2}D_2 + \cdots + \beta_{k}D_k,
\end{equation}
where $\beta_{i} = \alpha_{i,r}$ are the nonzero linear combination coefficients used in forming $G_r$ from data blocks.

Without loss of generality, we suppose that $k$ is a multiple of $p$ for easier explanation, and let $g = k/p$ be the number of data blocks in each local repair group. Then, we build local parity block $L_j$ (where $1 \le j \le p$) as
\begin{equation}\label{eq:xxx-2}
  L_j = \beta_{(j-1)g+1}D_{(j-1)g+1} + \cdots + \beta_{jg}D_{jg}.
\end{equation}

\paragraph{Single-node repair in CP-Azure LRC.}
The first step in repairing a single failed block is to identify its type. To be specific, we partition single-node failures into four types: (1) the failure of one data block, (2) the failure of one of the first $r-1$ global parity blocks, (3) the failure of the last global parity block, and (4) the failure of one local parity block. The repair method for each type proceeds as follows:
\begin{enumerate}[label=(\arabic*)]
  \item If a data block fails, we repair it within its local repair group by reading its local parity block and the other $g-1$ data blocks in its group; the repair bandwidth is $g$ blocks.
  \item If one of the first $r-1$ global parity blocks fails, we recompute it from the base MDS stripe with repair bandwidth of $k$ blocks.
  \item If the last global parity block $G_r$ fails, we read the $p$ local parity blocks and sum them; the bandwidth is $p$ blocks.
  \item If a local parity block $L_j$ fails, there are two equally simple choices: we can read its $g$ data blocks to recompute it, or use the identity~\eqref{eq:implied-repair} to read $G_r$ and the other $p-1$ local parity blocks; we pick the repair procedure with smaller repair bandwidth, i.e., $\min\{g,p\}$ blocks. In typical wide stripe settings where $g > p$, we adopt the second method to repair the local parity block $L_j$.
\end{enumerate}

\paragraph{Multi-node repair in CP-Azure LRC.}\label{sec:muti-node-cp-azure}
In CP-Azure LRC, multi-node repair follows a 'local-first, global-as-fallback' policy. Given a failure pattern, the repair procedure determines whether each failure can be handled within a local repair group or the cascaded group, and triggers global repair only when necessary. The following three cases summarize the decision process.
\begin{enumerate}[label=(\arabic*)]
  \item {\bf All groups contain at most one failure.} In this case, all failures can be repaired locally. Note that if two failures occur in the same group, but one is a local parity block, they are treated as belonging to {\em different} groups: one in the local repair group and the other in the cascaded group. For example, in Figure~\ref{fig:cp-lrcs}(b), if $D_1$ and $L_1$ fail simultaneously, $D_1$ is considered a failure in its local repair group, whereas $L_1$ is treated as a failure in the cascaded group. Thus, two independent single-node local repairs are performed.
  \item {\bf A group has more than one failure.} This situation cannot be resolved locally. The system must perform global repair by accessing $k$ surviving blocks. The failures within groups containing only one lost block are repaired locally, and the remaining failures are resolved through global repair. When repairing a local parity block, the reconstruction reuses data accessed during global repair to avoid redundant bandwidth.
  \item {\bf A failure occurs in one of the first $(r-1)$ global parity blocks.} Because these global parity blocks are outside the cascaded structure, their repair always triggers global repair. The failures are handled according to the rules in Case (2).
\end{enumerate}
Note that the maximum number of blocks accessed for multi-node repair is $k$, as the $k$ blocks selected for global repair already include blocks necessary for local repairs.

\paragraph{Fault tolerance analysis.} A $(k, r, p)$ CP-Azure LRC preserves the fault-tolerance of the underlying $(k,r)$ MDS code \cite{blaum2013partial,gopalan12locality,greenan2010mean} (Section~\ref{sec:base}), thus it tolerates any $r$ node failures.  However, it cannot tolerate arbitrary $r+1$ node failures. For example, if $r+1$ data blocks fail within one local repair group, the surviving $r$ global parity blocks and the local parity block are insufficient to decode the missing data blocks, because the local parity block is linearly dependent on the last global parity block under the cascaded construction. Thus, the minimum distance is exactly $r+1$. Furthermore, failure patterns of size $r+i$ (for $1 \le i \le p$) are still decodable as long as $i$ failures occur in $i$ distinct groups. In these cases, the $r+i$ erasures can be corrected using the $r$ global parity blocks with the $i$ available local parity blocks.

\paragraph{A $(6,2,2)$ CP-Azure LRC.} Figure~\ref{fig:cp-lrcs}(b) illustrates a $(6,2,2)$ CP-Azure LRC stripe, which starts from a $(6,2)$ MDS stripe shown in Figure~\ref{fig:cp-lrcs}(a). Then, we partition the six data blocks into two local repair groups, i.e., $(D_1, D_2, D_3)$ and $(D_4, D_5, D_6)$. Next, for each group, we generate a local parity block based on the coefficients of the linear combination from $G_2$. More precisely, we have
\begin{equation}\label{eq:L_1}
  L_1 = \alpha_{1,2}D_1 + \alpha_{2,2}D_2 +\alpha_{3,2}D_3,
\end{equation}
\begin{equation}
  L_2 = \alpha_{4,2}D_4 + \alpha_{5,2}D_5 +\alpha_{6,2}D_6,
\end{equation}\label{eq:L_2}
i.e., the coefficients for $G_2$ are decomposed and spread across $L_1$ and $L_2$. Obviously, $L_1 + L_2 = G_2$.

Examples of single-node repair:
(1) If a data block $D_1$ fails, we repair it by collecting $D_2,D_3,L_1$; the repair bandwidth is 3 blocks.
(2) If the first global parity block $G_1$ fails, we repair it by accessing all data blocks; the repair bandwidth is six blocks.
(3) If the last global parity block $G_2$ fails, we repair it within the cascaded group; the repair bandwidth is two blocks.
(4) If a local parity block (e.g., $L_1$) fails, we also repair it within the cascaded group, meaning the repair bandwidth is two blocks.

Examples of multi-node repair:
(1) If each group has no more than one failure, e.g., $ D_1$ and $ G_2$ fail, we collect $ D_2$, $D_3$, and $ L_1$, $L_2$ to repair $D_1$ and $G_2$, respectively; the repair bandwidth is four blocks.
(2) If there exits more than one failure in one group, e.g., $D_1, D_2, L_2$ fail, we collect $D_3, D_4, D_5, D_6, G_1, G_2$ for repair; note that $L_2$ uses $D_4, D_5, D_6$ that are accessed for global repair to save the cost; the repair bandwidth is six blocks.
(3) If one of the first $r - 1$ global parity blocks fails, e.g., $D_1, G_1$ fail, we collect $D_2, D_3, D_4, D_5, D_6, G_2$ to repair them simultaneously; the repair bandwidth is 6 blocks.

\subsection{CP-Uniform LRC} \label{sec:cp-uniform}

To build the local repair groups of a $(k,r,p)$ {CP-Uniform} LRC, we treat data blocks and certain global parity blocks {the same} during grouping. Concretely, we take all blocks except the last global parity block (i.e., the $k$ data blocks $D_1,\ldots, D_k$ and the first $r-1$ global parity blocks $G_1,\ldots, G_{r-1}$) and split them as evenly as possible into $p$ groups. Each group then produces one local parity block $L_i$ by linearly combining its own items (which may include both data and global parity blocks).

CP-Uniform follows the same central rule as CP-Azure:
\begin{equation}\label{eq:implied-repair-uniform}
  L_1 + L_2 + \cdots + L_p = G_r.
\end{equation}
To achieve this, we suppose that there exist $k + r -1$ nonzero coefficients $\gamma_i$ (where $1\le i\le k$) and $\eta_j$ (where $1\le j\le r-1$) in GF($2^w$) such that the base MDS code satisfies the following linear combination:
\begin{equation}\label{eq:combination}
  G_r = \gamma_1D_1 + \cdots + \gamma_kD_k + \eta_1G_1 + \cdots + \eta_{r-1}G_{r-1}.
\end{equation}
Note that in the appendix, we explicitly give these $k+r-1$ linear combination coefficients satisfying~\eqref{eq:combination} for the base $(k,r)$ Cauchy Reed-Solomon Codes.

Now we are ready to describe how to build the $p$ local parity blocks. Still, without loss of generality, we suppose that $k+r-1$ is a multiple of $p$ and let $h = (k+r-1)/p$ be the number of (data or parity) blocks in each local repair group. Then, each local parity block is a linear combination of the $h$ blocks within the corresponding group, with coefficients given by equation~\eqref{eq:combination}. Then, we can check that the local parity blocks and the last global parity block satisfy the equation~\eqref{eq:implied-repair-uniform}.

\paragraph{Single-node repair in CP-Uniform LRC.} The first step is to identify the failed block type. Here, we partition single-node failures into three types: (1) the failure of a data block or one of the first $r-1$ global parity blocks, (2) the failure of the last global parity block, and (3) the failure of one local parity block. The repair methods for CP-Uniform are:
\begin{enumerate}[label=(\arabic*)]
  \item If a data block or one of the first $r-1$ global parity blocks fails. Repair within its local repair group by reading the local parity block of that group and other $h-1$ items in the same group; the repair bandwidth is $h$ blocks.
  \item If the last global parity block fails. Use the identity in equation~\eqref{eq:implied-repair-uniform} and sum all $p$ local parity blocks; the repair bandwidth is $p$ blocks.
  \item If a local parity block fails. We have two choices: (a) read the $h$ items in its own group and recompute this local parity block (with repair bandwidth $h$ blocks), or (b) use equation~\eqref{eq:implied-repair-uniform} to read $G_r$ and the other $p-1$ local parity blocks (with repair bandwidth $p$ blocks). We choose the option with a smaller bandwidth, i.e., $\min\{h,p\}$ blocks. In typical wide stripe settings where $h > p$, we adopt the second method to repair the local parity block.
\end{enumerate}
As a final remark, when $k+r-1$ is not a multiple of $p$, group sizes differ by at most one; use the actual group size (either $\lfloor (k+r-1)/p\rfloor$ or $\lceil (k+r-1)/p\rceil$) in place of $h$ for the group-based repairs.

\paragraph{Multi-node repair in CP-Uniform LRC.}\label{sec:muti-node-cp-uni} CP-Uniform LRC adopts the same 'local-first, global-as-fallback' policy as CP-Azure LRC. Given a failure pattern, the system determines whether to perform local repairs only or invoke global repair based on the following two cases.
\begin{enumerate}[label=(\arabic*)]
  \item {\bf All groups contain at most one failure.} In this case, all failures are independent and can be repaired locally.
  \item {\bf A group has more than one failure.} Local repair is insufficient. The system triggers global repair by accessing $k$ surviving blocks. Failures in groups containing only one lost block are first repaired locally, while the remaining failures (those in overloaded groups) are repaired globally. Similar to CP-Azure LRC, the repair of a local parity block reuses data accessed during global repair to avoid redundant bandwidth.
\end{enumerate}
As in CP-Azure LRC, the number of accessed blocks never exceeds $k$, because the $k$ blocks selected for global repair already cover all data required by local repairs.

\paragraph{Fault tolerance analysis.} A $(k,r,p)$ CP-Uniform LRC also inherits the fault-tolerance of the underlying $(k,r)$ MDS code and therefore tolerates any $r$ arbitrary node failures. However, similar to CP-Azure LRC, it cannot tolerate arbitrary $r+1$ failures. Hence, the minimum distance remains $r+1$. Similarly, it can still correct $r+i$ failures (for $1 \le i \le p$) provided that the $i$ failures are distributed across $i$ distinct groups.

\paragraph{A $(6,2,2)$ CP-Uniform LRC.} Figure~\ref{fig:cp-lrcs}(c) illustrates a $(6,2,2)$ CP-Uniform LRC stripe, which also starts from a $(6,2)$ MDS stripe shown in Figure~\ref{fig:cp-lrcs}(a). Then, we partition the six data blocks and the first global parity block into two local repair groups, i.e., $(D_1, D_2, D_3)$ and $(D_4, D_5, D_6, G_1)$. Following that, for each group, we generate a local parity block according to the coefficients in the appendix. More precisely, we have
\begin{equation}\label{eq:L_1_}
  L_1 = \gamma_1D_1 + \gamma_2D_2 +\gamma_3D_3,
\end{equation}
\begin{equation}
  L_2 = \gamma_4D_4 + \gamma_5D_5 +\gamma_6D_6 + \eta_1G_1,
\end{equation}\label{eq:L_2_}
i.e., the coefficients for $G_2$ are decomposed and spread as evenly as possible across $L_1$ and $L_2$. Again, $L_1 + L_2 = G_2$.

Examples of single-node repair:
(1) If a data block $D_1$ fails, we can repair it by collecting $D_2,D_3,L_1$, and the repair bandwidth is 3 blocks. If the first global parity block $G_1$ fails, we repair it by accessing $D_4, D_5, D_6, L_1$, and the repair bandwidth is four blocks.
(2) If the last global parity block $G_2$ fails, $L_1$ and $L_2$ are utilized for repair within the cascaded group, incurring a repair bandwidth of 2 blocks.
(3) If a local parity block $L_1$ fails, $L_2$ and $G_2$ are utilized for repair, which also incurs a repair bandwidth of 2 blocks.

Examples of multi-node repair:
(1) If each group has no more than one failure, e.g., $D_1, G_2$ fail, we collect $D_2, D_3, L_1, L_2$ for repair; the bandwidth is four blocks.
(2) If there exists more than one failure in one group, e.g., $D_1, D_2, L_2$ fail, we collect $D_3, D_4, D_5, D_6, G_1, G_2$ for repair; the bandwidth is six blocks.

\subsection{Discussion} \label{subsec:design_discussion}

We now revisit the challenges outlined at the beginning of this section and discuss how CP-LRCs address them. A $(k,r,p)$ CP-LRC preserves the fault-tolerance of the underlying $(k,r)$ MDS code, while augmenting it with $p$ local parity blocks whose encoding coefficients are derived by decomposing those of the last global parity block. This enables the cascaded relationship between local and global parity blocks, which is key to achieving both efficient local repair and full MDS resiliency. CP-LRCs enable low-bandwidth single-node repair and employ a 'local-first, global-as-fallback' policy to handle multi-node failures. Under this policy, failures are repaired within their local repair groups or the cascaded group, whenever possible, and resort to global repair only when dictated by the failure pattern. Finally, the construction of CP-LRCs imposes no restrictions on code parameters, making the framework broadly applicable and generalizable.

A key design choice in CP-LRCs is to decompose the encoding coefficients of a global parity block and distribute them across all local parity blocks to form the cascaded group. Extending the cascaded relationship to more global parity blocks would introduce additional linear dependencies among parity equations, reducing distance and thus degrading fault tolerance. Limiting cascaded coupling to one global parity block is therefore essential to preserving the code's MDS-level fault tolerance.

We instantiate our framework with CP-Azure and CP-Uniform because Azure LRC provides strong fault tolerance \cite{hu21} while Uniform Cauchy LRC offers low repair cost \cite{kadekodi23} among existing wide stripe LRCs. However, our CP-LRCs can also be applied atop Azure LRC+1, Optimal Cauchy LRC, and other LRC variants by decomposing coefficients according to their respective encoding structures.
% !TEX root = ../main.tex

\section{Implementation}

We implement CP-LRCs in a distributed storage system prototype to demonstrate its practical benefits. The LRC approaches rely on Jerasure \cite{plank09}, while our prototype is realized in C++ with around 10,000 lines of code. Section~\ref{sec:system-architecture} details the system architecture, Section~\ref{sec:encoding-decoding} shows the encoding and decoding workflows, Section~\ref{sec:read-repair-optimization} presents our optimizations for read and repair operations, and Section~\ref{sec:metadata-management} gives the metadata management.

\subsection{System Architecture} \label{sec:system-architecture}

\begin{figure}[!t]
  \centering
  \includegraphics[width=0.85\linewidth]{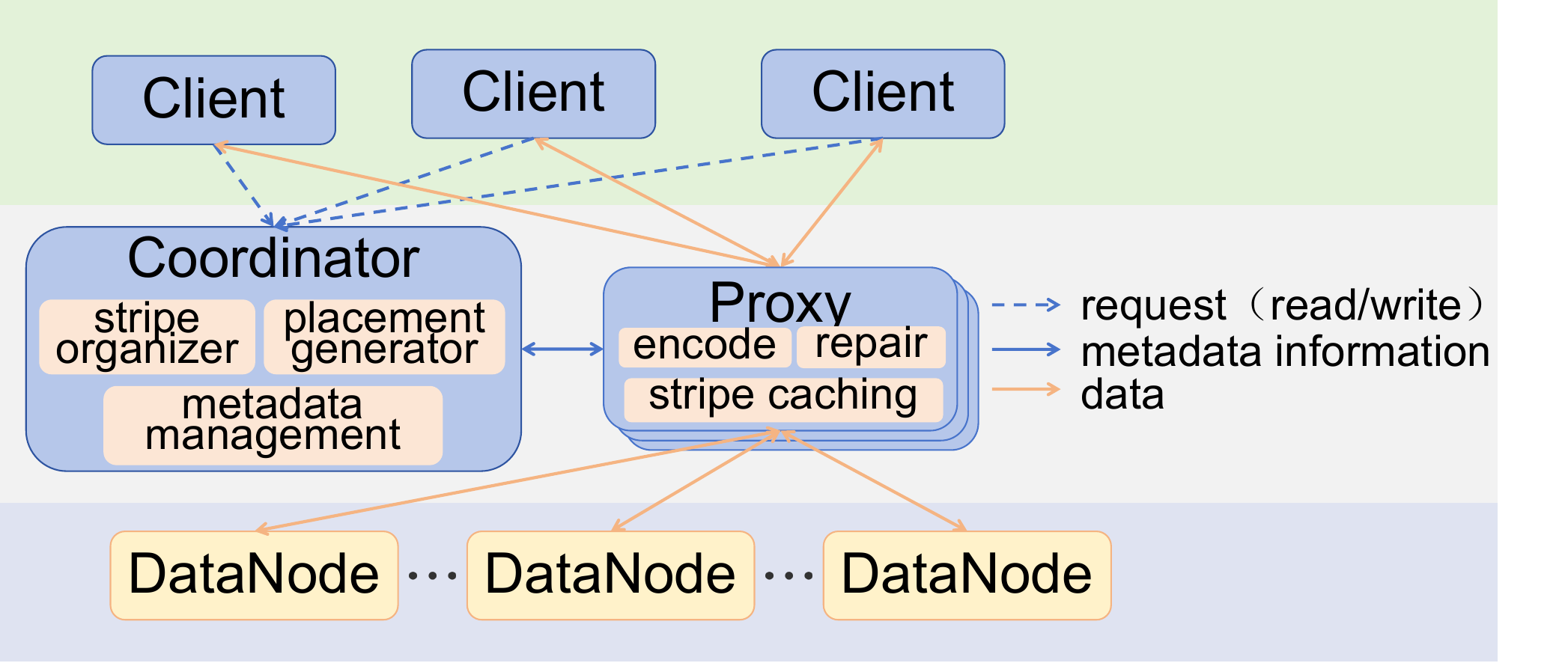}
  \caption{Distributed Prototype System}
  \label{fig:impl_arch}
\end{figure}

As shown in Figure~\ref{fig:impl_arch}, our prototype includes clients, a coordinator, proxies, and multiple data nodes. A client provides APIs for basic file operations such as read and write. The coordinator manages metadata, including stripe information (i.e., how the files constitute a stripe, the stripe-to-block mappings, and the coding parameters) and block information (i.e., the block locations). The proxies perform the actual encoding, decoding, and repair operations. The data nodes store data and parity blocks; each data node is associated with a unique identifier, IP address, and port number.

To flexibly tackle small files, our prototype organizes multiple small files into a single stripe; if the stripe is not fully occupied, we pad it with zeros.

\subsection{Encoding and Decoding Workflows} \label{sec:encoding-decoding}

Encoding is performed along the critical path of write operations, while decoding occurs during degraded read operations. The proxy manages both processes through well-defined, multi-stage workflows designed to ensure efficient data storage and reliable recovery.

Specifically, the encoding process consists of three key stages:
(1) Pre-encoding: Multiple small files are aggregated into a stripe. Each file's metadata records its size, stripe ID, and offset within the stripe.
(2) Parity generation: Based on the coding strategy (e.g., CP-Azure), local and global parity blocks are generated from the data blocks. To maintain uninterrupted data access during encoding, multiple temporary replicas are used to back up data and handle user requests concurrently.
(3) Data storage: Finally, the data and parity blocks are distributed across data nodes for persistent storage.

The decoding process comprises five stages:
(1) Repair triggering: Upon detecting node failures, the proxy initiates the repair workflow.
(2) Metadata retrieval: The proxy requests the corresponding metadata and repair plan from the coordinator. The coordinator determines whether to perform local or global repair according to the repair algorithms.
(3) Data collection: The proxy retrieves data from data nodes, precisely aligning the requested file data within the stripe to minimize read amplification and bandwidth consumption (Section~\ref{sec:read-repair-optimization}).
(4) Failure decoding: The proxy decodes and reconstructs the missing data.
(5) Data delivery: The reconstructed data is returned to the client, completing the degraded read operation.

\subsection{Read and Repair Optimizations} \label{sec:read-repair-optimization}

In wide stripe settings, a single stripe may contain numerous files. Therefore, the conventional block-level I/O access pattern causes significant I/O amplification during read and repair operations. To mitigate such amplification, our implementation performs fine-grained, file-level read and repair.

\begin{figure*}[t]
  \centering
  \begin{subfigure}[t]{0.3\linewidth}
    \centering
    \includegraphics[scale=0.60]{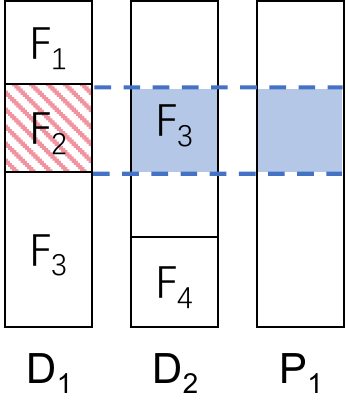}
    \caption{Access $F_2$ when $D_1$ fails}
    \label{fig:single_opt}
  \end{subfigure}
  \begin{subfigure}[t]{0.35\linewidth}
    \centering
    \includegraphics[scale=0.60]{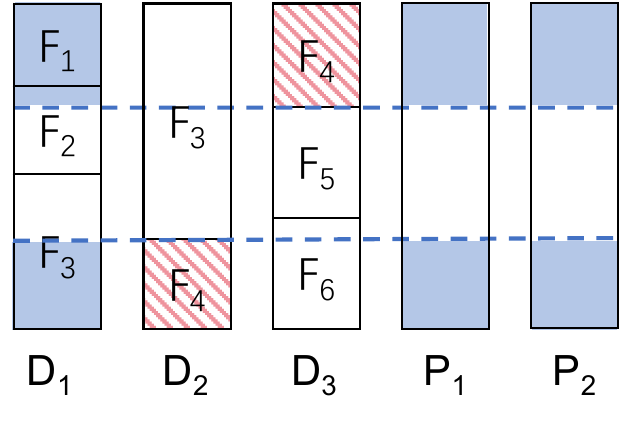}
    \caption{Access $F_4$ when $D_2$ and $D_3$ fail}
    \label{fig:double_repair}
  \end{subfigure}
  \begin{subfigure}[t]{0.3\linewidth}
    \centering
    \includegraphics[scale=0.60]{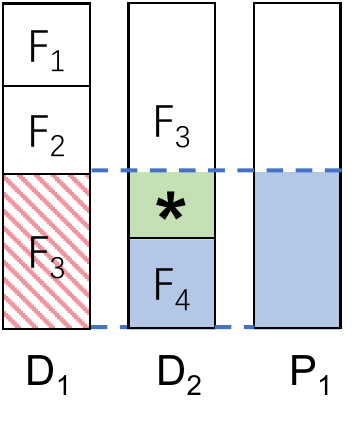}
    \caption{\parbox{\linewidth}{Eliminate repeated-read for accessing $F_3$ when $D_1$ fails. Where ``*'' denotes repeated read.}}
    \label{fig:single_repair_2}
  \end{subfigure}
  \caption{Illustration of file-level repair optimization strategies under a stripe-based layout.
    $D_i$ ($i=1,2,\ldots$) denotes a data block, and $P_i$ represents a parity block. $F_1, F_2, \dots, F_6$ are files (or file fragments) distributed across blocks. The red hatched regions indicate the failed data segments that trigger degraded reads. The blue shaded regions mark the portions of surviving blocks accessed during repair.}
  \label{fig:read_repair_optimization}
\end{figure*}

For scenarios where the requested file data resides within one or more blocks, instead of reading entire blocks, the proxy first obtains the file boundaries from the coordinator and then reads only the required file segments. This ensures that data retrieval precisely matches user requests.

\paragraph{Repair Optimization.} Similarly, during degraded reads, the proxy performs selective data access by reading only the necessary segments from the surviving blocks for decoding. For example, in Figure~\ref{fig:read_repair_optimization}(a), when the client accesses $F_2$ residing in a failed block $D_1$, the proxy avoids fetching entire surviving blocks $D_2$ and $P_1$. Instead, it retrieves only the segments that align with the boundaries of $F_2$, minimizing unnecessary reads. In a more complex case shown in Figure~\ref{fig:read_repair_optimization}(b), file $F_4$ spans two failed blocks $D_2$ and $D_3$. Here, the proxy again reads only the file-aligned portions of surviving blocks required for double-node decoding.

When a file spans multiple blocks, the retrieved segments from surviving blocks may partially overlap with the file's existing data. For example, in Figure~\ref{fig:read_repair_optimization}(c), file $F_3$ spans blocks $D_1$ and $D_2$. If $D_1$ fails, the decoding segment retrieved from $D_2$ overlaps with the portion of $F_3$ already read from $D_2$. In such a case, the proxy eliminates repeated reads by skipping overlapping data and fetching only the non-overlapping portions from surviving blocks. This reduces both network bandwidth and disk read overhead during degraded reads.

\subsection{Metadata Management} \label{sec:metadata-management}

The coordinator maintains metadata using four compact indexes: stripe, block, object, and node indices. The stripe index records each stripe's globally unique stripe\_id, parameters $(k,r,p)$, coding strategy, and the node\_id list for all data blocks, local parity blocks, and global parity blocks within this stripe. The block index records each block's id (using a composite key of stripe\_id plus block index) and a linked list of files stored within this block. The object index stores file-specific metadata, including the file's id, size, associated stripe\_id, block indexes, and offsets within blocks. The node index tracks physical nodes, recording each node's node\_id, IP, port, and liveness status. This hierarchical indexing enables flexible mapping from logical objects to physical blocks and supports efficient read/repair coordination with minimal metadata overhead.

To evaluate metadata storage cost, consider a 100 GB dataset with a block size of 2 MB and parameters $(n, k) = (8, 6)$, where the average file size is 128 KB. As the stripe index requires 128  bytes per stripe, the block index 64 bytes per block, and the object index 32 bytes per object, the total metadata size for these three indexes amounts to approximately $1.04 + 4.36 + 25.00 = 30.4$ MB. This corresponds to roughly 0.03\% of the total data volume. Among these, the object index dominates the metadata footprint due to the large number of small files. The node index remains extremely lightweight (only several KB for $n=8$), yet plays a critical role in tracking node states and network locations. Overall, the lightweight hash-index design achieves low memory overhead.
% !TEX root = ../main.tex

\section{Experiments} \label{sec:evaluation}

We evaluate our proposed CP-LRCs through both theoretical analysis and cloud evaluation. We first theoretically compare CP-LRCs with existing wide stripe LRCs using the metrics defined in Section~\ref{sec:wide-lrc}. Then, we deploy our distributed storage prototype in Alibaba Cloud and evaluate the actual cloud performance of CP-LRCs under various failure scenarios.

\paragraph{Summary of results.}
From theoretical analysis, CP-LRCs reduce the $\text{ARC}_1$ and $\text{ARC}_2$ of the baseline LRCs by up to 47.5\% and 19.9\%, respectively, and improve the MTTDL of the baselines by up to 105.3\%. From cloud experiments, CP-LRCs reduce the single-node and two-node repair times of the baseline LRCs by up to 41\% and 26\%, respectively. The file-level repair optimization strategies cut the degraded read latency by an average of 58.6\% for small-size files.

\subsection{Theoretical Analysis} \label{sec:theoretical-evaluation}

\subsubsection{Repair Bandwidth} \label{sec:repair-bandwidth}

\begin{table}[!t]
    \centering
    \caption{Evaluation Parameters.}
    \begin{tabular}{cccccr}
        \toprule
        \textbf{Label} & $k$ & $r$ & $p$ & \textbf{Code Rate} ($k/(k+r+p)$) \\
        \midrule
        P1             & 6   & 2   & 2   & 6/10 = 0.600                     \\
        P2             & 12  & 2   & 2   & 12/16 = 0.750                    \\
        P3             & 16  & 3   & 2   & 16/21 = 0.762                    \\
        P4             & 20  & 3   & 5   & 20/28 = 0.714                    \\
        P5             & 24  & 2   & 2   & 24/28 = 0.857                    \\
        P6             & 48  & 4   & 3   & 48/55 = 0.873                    \\
        P7             & 72  & 4   & 4   & 72/80 = 0.900                    \\
        P8             & 96  & 5   & 4   & 96/105 = 0.914                   \\
        \bottomrule
    \end{tabular}
    \label{tab:parameters}
\end{table}

We compare the ADRC, $\text{ARC}_1$, $\text{ARC}_2$ values of CP-LRCs with Azure LRC, Azure LRC+1, Optimal Cauchy LRC, Uniform Cauchy LRC under eight sets of parameters $(k, r, p)$ (denoted P1-P8 as shown in Table~\ref{tab:parameters}). Note that P5, P6, P7, and P8 are wide stripe parameters adopted in Google \cite{kadekodi23}. Table~\ref{tab:all_schemes_comparison} shows the results. We summarize the observations below.
\begin{itemize}[leftmargin=*]
    \item ADRC: CP-Azure achieves the smallest ADRC as Azure, while CP-Uniform lowers ADRC over Uniform due to the benefit of shortening the size of original local repair groups. Also, both CP-Azure and CP-Uniform consistently outperform Azure LRC+1 and Optimal LRC baselines.
    \item ARC$_1$: CP-Azure and CP-Uniform have the smallest and second smallest ARC$_1$ across all parameters due to the low cost of parity repair. For example, under $(6, 2, 2)$, CP-Azure and CP-Uniform reduce ARC$_1$ of other LRCs by 16.7\%-40\% and 13.9\%-38\%, respectively; under $(24, 2, 2)$, CP-Azure and CP-Uniform reduce ARC$_1$ of other LRCs by 11.7\%-47.5\% and 11.4\%-47.4\%.
    \item ARC$_2$: Under double-node failures, CP-Azure and CP-Uniform also have the smallest and second smallest ARC$_2$. The reason is that CP-LRCs are more likely to perform local repairs than global repairs. CP-LRCs consistently outperform other LRCs under both narrow-stripe and wide-stripe parameters.
\end{itemize}

\begin{table}[!t]
    \centering
    \caption{A comprehensive comparison of theoretical repair costs across different LRC constructions.}
    \label{tab:all_schemes_comparison}
    \small
    \setlength{\tabcolsep}{3pt}
    \resizebox{\columnwidth}{!}{%
        \begin{tabular}{@{}l*{8}{c}@{}}
            \toprule
            \textbf{parameters} & P$_1$         & P$_2$          & P$_3$          & P$_4$          & P$_5$          & P$_6$          & P$_7$          & P$_8$          \\
            \midrule
            \multicolumn{9}{c}{\textbf{Average Degraded Read Cost (ADRC)}}                                                                                             \\
            \midrule
            Azure LRC           & \textbf{3.00} & \textbf{6.00}  & \textbf{8.00}  & \textbf{4.00}  & \textbf{12.00} & \textbf{16.00} & \textbf{18.00} & \textbf{24.00} \\
            Azure LRC+1         & 6.00          & 12.00          & 16.00          & 5.00           & 24.00          & 24.00          & 24.00          & 32.00          \\
            Optimal LRC         & 5.00          & 8.00           & 10.00          & 7.00           & 14.00          & 20.00          & 22.00          & 29.00          \\
            Uniform LRC         & 4.00          & 7.00           & 9.50           & 4.60           & 13.00          & 17.29          & 19.00          & 25.22          \\
            CP-Azure            & \textbf{3.00} & \textbf{6.00}  & \textbf{8.00}  & \textbf{4.00}  & \textbf{12.00} & \textbf{16.00} & \textbf{18.00} & \textbf{24.00} \\
            CP-Uniform          & 3.50          & 6.50           & 9.00           & 4.40           & 12.50          & 17.00          & 18.75          & 25.00          \\
            \midrule
            \multicolumn{9}{c}{\textbf{Average Single-node Repair Cost (ARC$_1$)}}                                                                                     \\
            \midrule
            Azure LRC           & 3.60          & 6.75           & 9.14           & 5.71           & 12.86          & 18.33          & 20.70          & 27.43          \\
            Azure LRC+1         & 4.80          & 10.13          & 13.52          & 4.71           & 21.64          & 22.18          & 22.75          & 30.46          \\
            Optimal LRC         & 5.00          & 8.00           & 11.00          & 7.00           & 13.00          & 20.00          & 22.00          & 29.00          \\
            Uniform LRC         & 4.00          & 7.00           & 9.52           & 4.64           & 13.00          & 17.35          & 19.00          & 25.22          \\
            CP-Azure            & \textbf{3.00} & \textbf{5.63}  & \textbf{7.90}  & 5.36           & \textbf{11.36} & 16.80          & 19.15          & 25.79          \\
            CP-Uniform          & 3.10          & 5.68           & 8.00           & \textbf{4.57}  & 11.39          & \textbf{15.98} & \textbf{17.84} & \textbf{24.00} \\
            \midrule
            \multicolumn{9}{c}{\textbf{Average Two-node Repair Cost (ARC$_2$)}}                                                                                        \\
            \midrule
            Azure LRC           & 6.00          & 12.00          & 16.00          & 12.06          & 24.00          & 38.66          & 47.32          & 63.03          \\
            Azure LRC+1         & 6.22          & 12.02          & 16.04          & 11.24          & 24.07          & 44.63          & 52.54          & 70.43          \\
            Optimal LRC         & 6.27          & 12.46          & 16.22          & 12.26          & 25.17          & 39.35          & 47.06          & 62.62          \\
            Uniform LRC         & 6.22          & 12.02          & 16.01          & 11.11          & 24.07          & 38.96          & 46.18          & 61.56          \\
            CP-Azure            & \textbf{5.47} & \textbf{10.68} & \textbf{14.30} & \textbf{10.63} & \textbf{21.82} & \textbf{35.73} & 43.88          & 59.43          \\
            CP-Uniform          & 5.80          & 10.99          & 14.37          & 10.64          & 22.03          & 35.86          & \textbf{42.98} & \textbf{58.15} \\
            \bottomrule
        \end{tabular}
    }
\end{table}

\subsubsection{Portion of Local Repair in Multi-Node Failures} \label{sec:local-repair-portion}
To quantify how effectively CP-LRCs enhance local repair under multi-node failures, we define the {\em portion of local repair} as the fraction of failure cases that can be repaired within local repair groups out of all multi-node failure combinations. Specifically, for the case of two-node failures, we enumerate all possible pairs of failed nodes and determine, for each pair, whether the lost blocks can be repaired using local repair groups or require global repair.

Table~\ref{tab:portion-local-repair} shows the portion of local repair under two-node failures across various LRCs. As shown, CP-Uniform consistently achieves the highest local repair portion, demonstrating its strong resilience to multi-node failures. CP-Azure also consistently outperforms traditional Azure and Azure+1. Under narrow stripe parameters, CP-Azure surpasses both the Optimal and Uniform schemes, indicating that locality is more impactful when each stripe contains fewer blocks. However, as the stripe width increases, CP-Azure’s advantage diminishes, as wider stripes incur a higher global repair possibility that mitigates locality benefits.

\begin{table}[t]
    \centering
    \caption{Portion of local repair under two-node failures.}
    \label{tab:portion-local-repair}
    \resizebox{\linewidth}{!}{
        \begin{tabular}{lcccccccc}
            \toprule
            \textbf{Method} & P$_1$         & P$_2$         & P$_3$         & P$_4$         & P$_5$         & P$_6$         & P$_7$         & P$_8$         \\
            \midrule
            Azure LRC       & 0.36          & 0.41          & 0.39          & 0.66          & 0.45          & 0.58          & 0.67          & 0.69          \\
            Azure LRC+1     & 0.47          & 0.33          & 0.32          & \textbf{0.83} & 0.20          & 0.59          & 0.71          & 0.71          \\
            Optimal LRC     & 0.62          & 0.61          & 0.62          & 0.82          & 0.57          & 0.71          & 0.78          & 0.77          \\
            Uniform LRC     & 0.56          & 0.53          & 0.52          & \textbf{0.83} & 0.52          & 0.70          & 0.76          & 0.76          \\
            CP-Azure        & 0.67          & 0.63          & 0.55          & 0.78          & 0.58          & 0.65          & 0.73          & 0.72          \\
            CP-Uniform      & \textbf{0.80} & \textbf{0.70} & \textbf{0.66} & \textbf{0.83} & \textbf{0.62} & \textbf{0.75} & \textbf{0.79} & \textbf{0.78} \\
            \bottomrule
        \end{tabular}}
\end{table}

In some scenarios, local repair may cost more than global repair, since it requires accessing both data blocks and global parity blocks. In other cases, local repair incurs the same cost as global repair and provides no advantage. To distinguish actual benefits, we define the {\em portion of effective local repair} as the fraction of multi-node failures where local repair achieves strictly lower cost than global repair.

Table~\ref{tab:optimized-local-repair} summarizes the results. Here, conventional LRCs exhibit nearly zero effective local repair capability under narrow-stripe settings, indicating that most failures must trigger global repair. In contrast, CP-LRCs maintain 20\%-50\% effective local repair coverage even in such restrictive settings, and this advantage further grows in wide stripe scenarios.

\begin{table}[t]
    \centering
    \caption{Portion of effective local repair (cost lower than global repair) under two-node failures.}
    \label{tab:optimized-local-repair}
    \resizebox{\linewidth}{!}{
        \begin{tabular}{lcccccccc}
            \toprule
            \textbf{Method} & P$_1$         & P$_2$         & P$_3$         & P$_4$         & P$_5$         & P$_6$         & P$_7$         & P$_8$         \\
            \midrule
            Azure LRC       & \gray{0.00}   & \gray{0.00}   & \gray{0.00}   & 0.66          & \gray{0.00}   & 0.58          & 0.67          & 0.69          \\
            Azure LRC+1     & \gray{0.00}   & \gray{0.00}   & \gray{0.00}   & 0.83          & \gray{0.00}   & 0.17          & 0.71          & 0.71          \\
            Optimal LRC     & \gray{0.00}   & \gray{0.00}   & \gray{0.00}   & 0.82          & \gray{0.00}   & 0.71          & 0.78          & 0.77          \\
            Uniform LRC     & \gray{0.00}   & \gray{0.00}   & \gray{0.00}   & 0.83          & \gray{0.00}   & 0.70          & 0.76          & 0.76          \\
            CP-Azure        & 0.47          & 0.33          & 0.24          & 0.78          & 0.20          & 0.73          & 0.73          & 0.72          \\
            CP-Uniform      & \textbf{0.53} & \textbf{0.35} & \textbf{0.27} & \textbf{0.83} & \textbf{0.21} & \textbf{0.79} & \textbf{0.79} & \textbf{0.78} \\
            \bottomrule
        \end{tabular}}
\end{table}

Overall, these results confirm that CP-LRCs significantly increase the local repair probability under multi-node failures. This improvement directly translates to reduced ARC$_2$ and shorter degraded durations, which collectively contribute to enhanced system reliability and higher MTTDL, as further analyzed in the next section.

\subsubsection{MTTDL Analysis} \label{sec:mttdl-analysis}

Wide stripe LRCs generally lower system reliability (MTTDL) due to increased repair bandwidth and a higher probability of data loss. In contrast, CP-LRCs enhance reliability by simultaneously lowering repair costs and increasing the likelihood of local repair in multi-failure scenarios. We now compute the MTTDL values for all wide stripe LRC schemes across various parameters. Table~\ref{tab:mttdl-results-transposed-single} presents the results, where a larger MTTDL indicates higher reliability.

\begin{table}[t]
    \centering
    \caption{MTTDL comparison across different LRC constructions.}
    \label{tab:mttdl-results-transposed-single}
    \setlength{\tabcolsep}{5pt}
    \renewcommand{\arraystretch}{1.1}
    \resizebox{\linewidth}{!}{%
        \begin{tabular}{lcccccccc}
            \toprule
            \textbf{Method} & \textbf{P$_1$}   & \textbf{P$_2$}   & \textbf{P$_3$}   & \textbf{P$_4$}   & \textbf{P$_5$}   & \textbf{P$_6$}   & \textbf{P$_7$}   & \textbf{P$_8$}   \\
            \midrule
            Azure LRC       & 2.66e17          & 4.67e11          & 1.62e14          & 3.05e27          & 1.90e14          & 1.38e21          & 2.50e22          & 5.32e23          \\
            Azure LRC+1     & 1.99e17          & 3.11e11          & 1.09e14          & 3.70e27          & 1.13e14          & 1.14e21          & 2.28e22          & 4.79e23          \\
            Optimal LRC     & 1.91e17          & 3.94e11          & 1.35e14          & 2.49e27          & 1.89e14          & 1.15e21          & 2.36e22          & 5.04e23          \\
            Uniform LRC     & 2.39e17          & 4.50e11          & 1.56e14          & 3.75e27          & 1.89e14          & 1.46e21          & 2.73e22          & 5.79e23          \\
            CP-Azure        & \textbf{3.19e17} & \textbf{5.60e11} & \textbf{1.88e14} & 3.25e27          & 2.16e14          & 1.50e21          & 2.71e22          & 5.66e23          \\
            CP-Uniform      & 3.09e17          & 5.55e11          & 1.85e14          & \textbf{3.81e27} & \textbf{2.32e14} & \textbf{1.58e21} & \textbf{3.12e22} & \textbf{6.55e23} \\
            \bottomrule
        \end{tabular}%
    }
\end{table}

Across all configurations, CP-Azure and CP-Uniform consistently achieve the highest and second-highest MTTDL. For example, under $(6,2,2)$, CP-Azure and CP-Uniform improve the MTTDL of other LRCs by 19.9\%-67\% and 16.2\%-61.8\%; under $(24,2,2)$, CP-Azure and CP-Uniform improve the MTTDL of other LRCs by 13.7\%-91.2\% and 22.1\%-105.3\%, respectively.

\subsection{Cloud Evaluation} \label{sec:practical-evaluation}

\subsubsection{Setup} \label{sec:setup}

We conduct experiments on Alibaba Cloud \cite{alibaba25}. We deploy our prototype on 18 instances in different zones in Beijing. Specifically, we deploy one \texttt{ecs.g6.8xlarge} instance in Zone I to act as the proxy. This instance is equipped with 32 vCPU and 128 GiB memory and runs Ubuntu 24.04 64-bit. We deploy 2 \texttt{ecs.r6e.xlarge} instances in Zone J to act as the client and coordinator, respectively.We deploy 15 \texttt{ecs.r6e.xlarge} instances in Zones K, J, and L to act as the datanodes. Each instance is equipped with 4 vCPU and 32 GiB of memory and runs Ubuntu 24.04 64-bit.

We adopt the same sets of parameters as in our analysis (i.e., P1-P8 in Table~\ref{tab:parameters}). The default configurations are set as follows: (1) Coding parameters: P5, i.e., (24, 2, 2); (2) Block size: 64 MB; (3) Bandwidth: 1 Gbps; (4) Stripe number: 10, such that the total storage volume is around 20 GB.

We measure both the single-node and two-node repair time. For single-node failures, we repair the failed block in each stripe in turn and record the average repair time. For two-node failures, we randomly trigger 10 patterns of two failed blocks per stripe and apply the same pattern to each LRC scheme; we then compute the average repair time. We average the results of each experiment over ten runs.

\subsubsection{Experiment 1: Single-Node Repair Under Different Parameters} \label{sec:exp1-single-repair}

\begin{figure*}[htbp]
    \centering
    \includegraphics[width=0.95\textwidth]{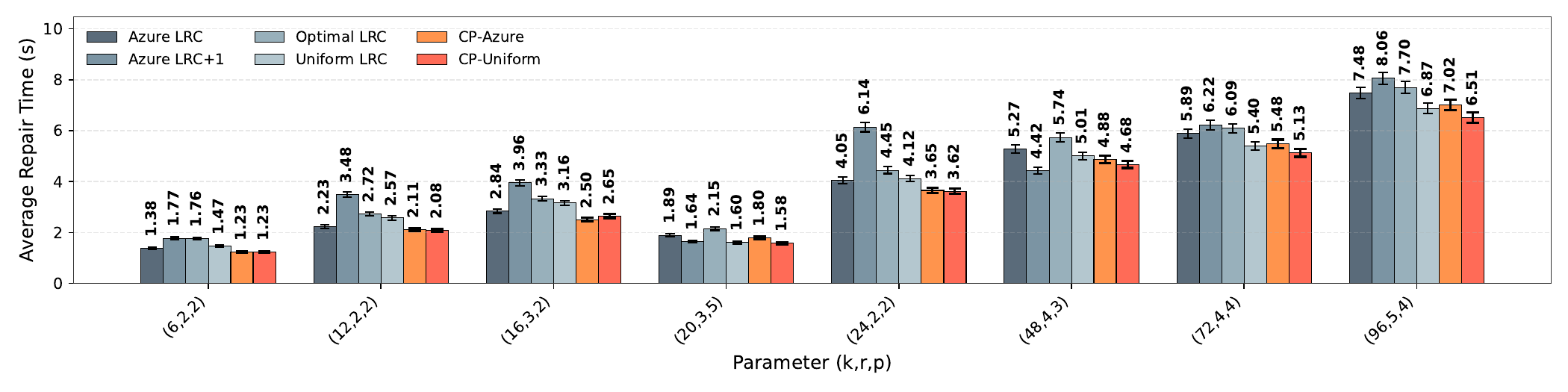}
    \caption{Experiment 1: Single Node Repair Under Different Parameters (P1, P2, $\dots$, P8). We plot the average repair time (in seconds) with error bars for single-node failures within a stripe. CP-Azure and CP-Uniform correspond to the smallest and second-smallest repair time across all settings. }
    \label{fig:exp_1_1}
\end{figure*}

In this experiment, we measure the repair time consumption under different parameters when a single block fails. The parameter is varied from P1 to P8 (Table~\ref{tab:parameters}), and the block size and network bandwidth are set to their default values. The results are shown in Figure~\ref{fig:exp_1_1}.

From Figure~\ref{fig:exp_1_1}, CP-LRCs consistently reduce the single-node repair time of traditional wide stripe LRCs across all parameter settings. For example, under $(12, 2, 2)$, CP-Azure and CP-Uniform reduce the single-node repair time of other LRCs by 5.4\%-39.4\% and 6.7\%-40.2\%, respectively; under $(24, 2, 2)$, CP-Azure LRC and CP-Uniform LRC reduce the repair time of other LRCs by 9.9\%-40.6\% and 10.6\%-41\%. The performance gains are particularly pronounced in wide stripe configurations, where CP-LRCs achieve up to 41\% reduction in repair time compared to Azure LRC+1 (P5). The results validate our theoretical analysis (Section~\ref{sec:repair-bandwidth}) and demonstrate that CP-LRCs effectively reduce the repair bandwidth consumption, leading to faster repair time.

\subsubsection{Experiment 2: Single-Node Repair Under Different Block Sizes} \label{sec:exp2-block-size}

\begin{figure*}[htbp]
    \centering
    \includegraphics[width=0.95\textwidth]{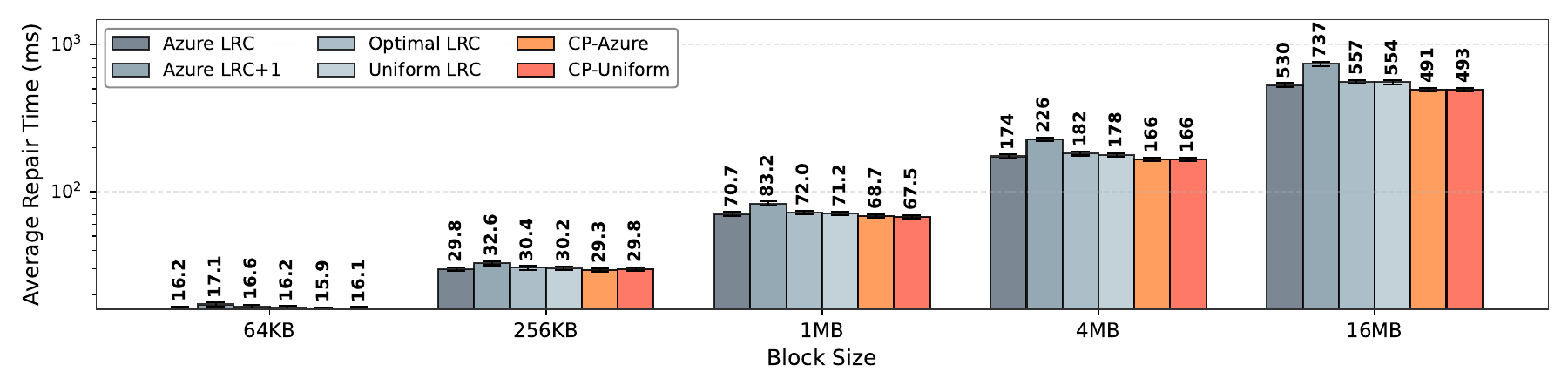}
    \caption{Experiment 2: Single node repair under different block sizes (64 KB to 16 MB). We plot the average repair time (in milliseconds) with error bars for single-node failures within a stripe. CP-Azure and CP-Uniform have the smallest and second-smallest repair times, respectively.}
    \label{fig:exp_2}
\end{figure*}

\begin{figure}
    \centering
    \includegraphics[width=0.9\linewidth]{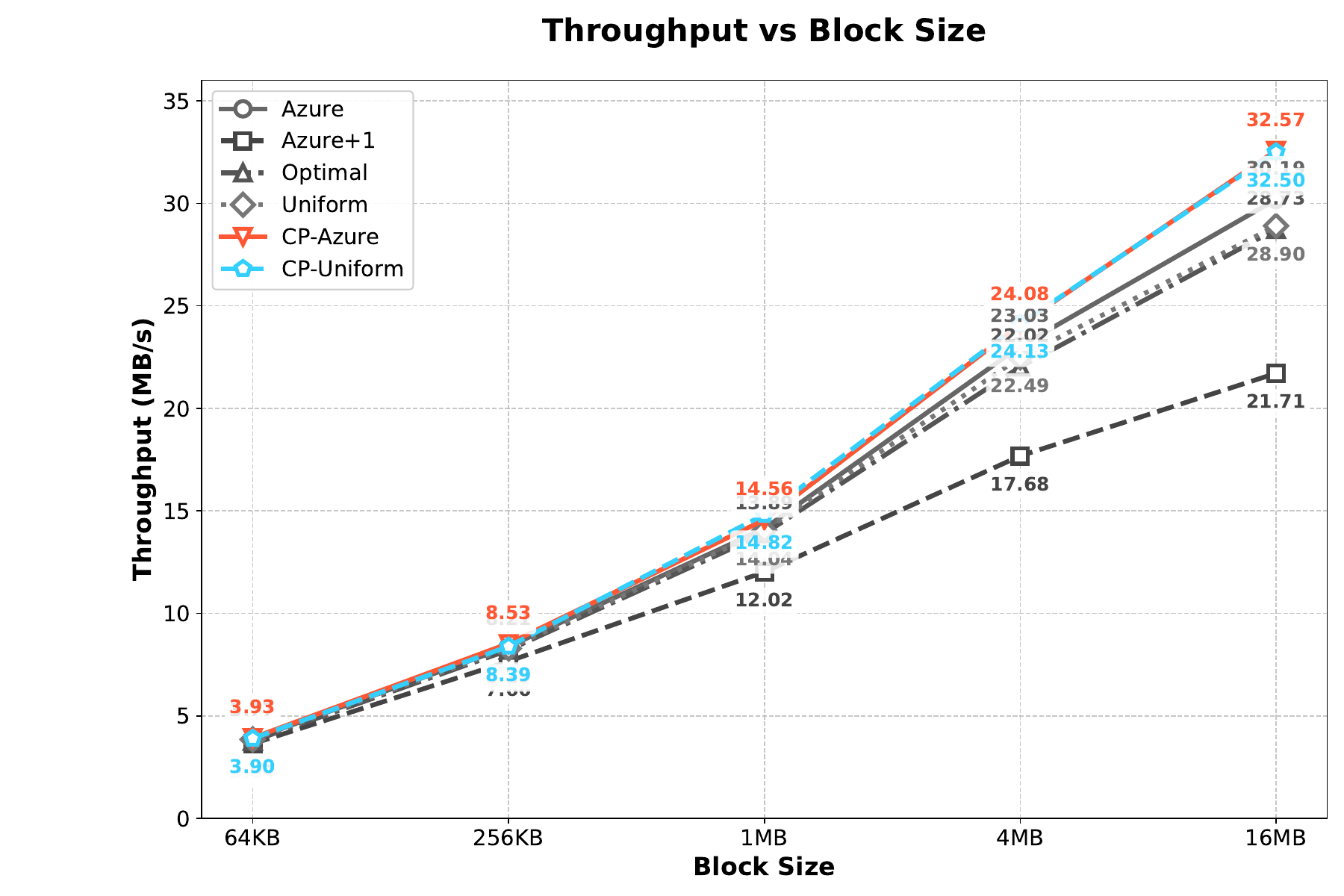}
    \caption{Single-node repair throughput (in MB/s) of different LRCs with varied block sizes (64 KB to 16 MB).}
    \label{fig:single-repair}
\end{figure}

We further investigate the impact of block size on single-node repair efficiency. We vary the block size from 64 KB to 16 MB and measure repair time, with the default values for parameters,  network bandwidth, and stripe number. Figure~\ref{fig:exp_2} shows the results. The repair time of all LRC schemes increases with the block size, as a larger data volume must be transferred over the network for repair. Nevertheless, CP-Azure and CP-Uniform consistently outperform all other LRCs: CP-Azure LRC reduces the single-node repair time of Azure LRC, Azure LRC+1, Optimal LRC, and Uniform LRC by 1.9\%-7.4\%, 7\%-33.3\%, 4.2\%-11.8\%, and 1.9\%-11.3\%. CP-Uniform LRC reduces the single-node repair time of Azure LRC, Azure LRC+1, Optimal LRC, and Uniform LRC by 0.6\%-7.0\%, 5.8\%-33.3\%, 3\%-11.5\%, and 0.6\%-11.0\%. CP-LRCs achieve larger gains with a larger block size.

Figure~\ref{fig:single-repair} further shows the single-node repair throughput. All LRC schemes achieve higher throughput with a larger block size, while CP-Azure and CP-Uniform consistently achieve the highest and second-highest throughput. CP-Azure LRC improves the single-node repair throughput of Azure LRC, Azure LRC+1, Optimal LRC, and Uniform LRC by 1.6\%-7.9\%, 8.0\%-50.0\%, 3.9\%-13.4\%, and 2.1\%-12.7\%. CP-Uniform LRC improves the single-node repair throughput of Azure LRC, Azure LRC+1, Optimal LRC, and Uniform LRC by 0.1\%-7.6\%, 7.1\%-49.8\%, 2.2\%-13.1\%, and 1.2\%-12.5\%.

\subsubsection{Experiment 3: Multi-Node Repair Under Different Parameters} \label{sec:exp3-multi-repair}

We now measure the repair time consumption under different parameters when two blocks fail simultaneously within a stripe. We adopt P1-P8, and the block size, network bandwidth, and stripe number are set as default. Figure~\ref{fig:double-repair_1} shows the results.

\begin{figure*}[htbp]
    \centering
    \includegraphics[width=0.95\textwidth]{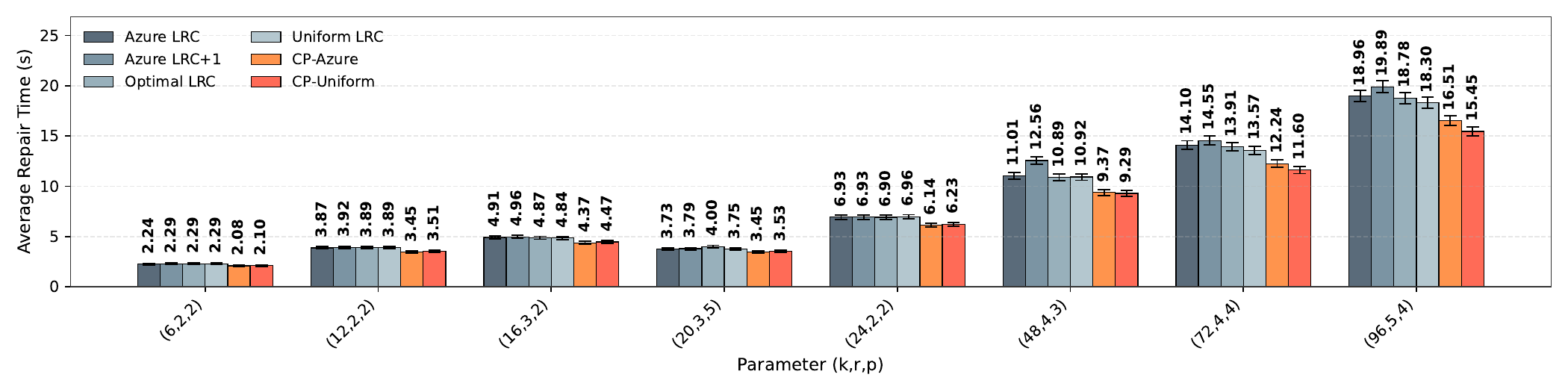}
    \caption{Experiment 3: Double Node Repair Under Different Parameters (P1, P2, $\dots$, P8). We plot the average repair time (in seconds) with error bars for two-node failures within a stripe. CP-Azure and CP-Uniform have the smallest and second-smallest repair times across all settings. }
    \label{fig:double-repair_1}
\end{figure*}

From Figure~\ref{fig:double-repair_1}, CP-LRCs consistently reduce the two-node repair time of traditional wide stripe LRCs across all parameter settings. For example, under $(12, 2, 2)$, CP-Azure and CP-Uniform reduce the two-node repair time of other LRCs by 10.9\%-12\% and 9.3\%-10.5\%, respectively; under $(48, 4, 3)$, CP-Azure LRC and CP-Uniform LRC reduce the repair time of other LRCs by 14.2\%-25.4\% and 14.9\%-26\%. Once again, CP-LRCs achieve larger gains under wide stripe configurations; for example, CP-LRCs achieve up to 26\% reduction in two-node repair time compared to Azure LRC+1 (P6). From our analysis, CP-LRCs increase the proportion of effective local repair compared to other LRCs. To validate this analysis, we further examine the effective local repair cases (where local repair is less costly than global repair) in this experiment. Under P6, CP-Azure LRC has 0.73 effective local repair cases, CP-Uniform LRC has 0.79, while Azure LRC, Azure LRC+1, Optimal LRC, and Uniform LRC have 0.58, 0.17, 0.71, and 0.7, respectively.

\subsubsection{Experiment 4: File-Level Repair Optimization Under Real-World Traces} \label{sec:exp4-trace}

\begin{figure}
    \centering
    \includegraphics[width=0.9\linewidth]{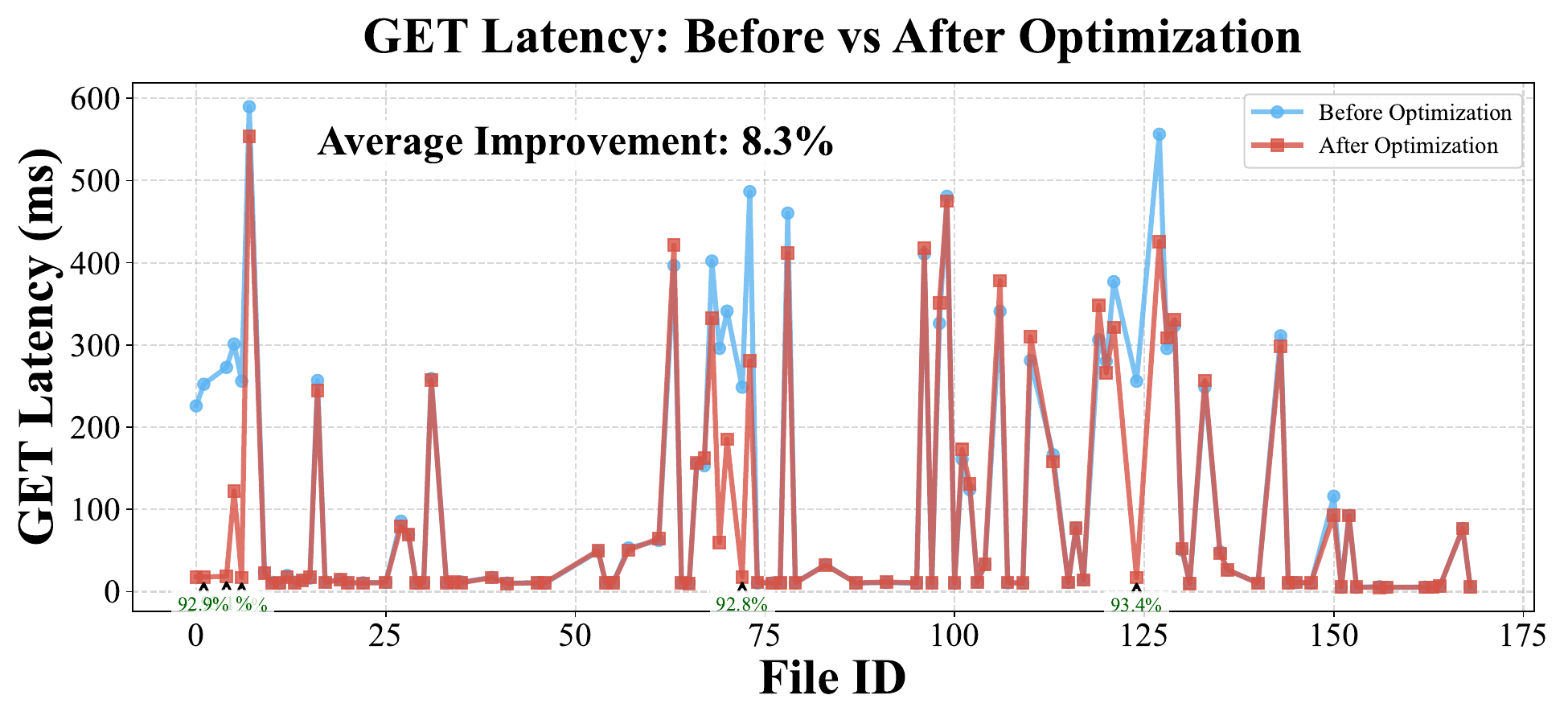}
    \caption{Experiment 4: Repair optimization under Facebook trace.}
    \label{fig:degraded-read}
\end{figure}

In this experiment, we test our file-level repair optimization for mitigating I/O amplification during degraded reads (Section~\ref{sec:read-repair-optimization}). We deploy the {\em FB-2010} trace, a large-scale production workload from Facebook comprising 30 million jobs processing 1.5 EB of data over 1.5 months \cite{chen12}. We randomly sample 100 files with sizes ranging from 5 KB to 30 MB, and encode them using Azure LRC with a block size of 16 MB. During trace execution, we simulate single-node failures and measure the degraded (or normal) read latency for each file. Figure~\ref{fig:degraded-read} shows the results.

From Figure~\ref{fig:degraded-read}, our file-level optimization effectively reduces the degraded read latency by alleviating unnecessary block-level reads. Across all files, the optimization reduces the average degraded read time from 120.4 ms to 96.5 ms, resulting in a 19.8\% improvement. The benefit is more pronounced for small files (size $<$ 1 MB), where the latency drops from 41.1 ms to 17.0 ms, a 58.6\% reduction. For medium and large files, the optimization yields improvements of 19.5\% and 5.6\%, respectively.
% !TEX root = ../main.tex

\section{Related Work}

\subsection{New Constructions of Locally Repairable Codes}

Locally Repairable Codes (LRCs) have emerged as a critical mechanism for reducing the high repair bandwidth in erasure-coded storage systems. The theoretical work by Tamo and Barg~\cite{tamo14} established a family of LRCs that achieve the Singleton-type bound with minimum locality, providing a basis for subsequent practical designs. Recent industrial efforts have developed new LRC variants that balance repair efficiency, fault tolerance, and deployment practicality. For example, Kolosov et al. \cite{kolosov20} developed several new LRCs (e.g., Azure LRC+1) and deployed them in Ceph across multiple Availability Zones. Google's work \cite{kadekodi23} designed Optimal and Uniform Cauchy LRCs with balanced fault tolerance and uniform repair bandwidth, while also enabling maintenance-robust-efficient placement strategies. Concurrently, the studies \cite{wu2020, hu21} explored hierarchical repair structures and parity layouts to minimize the cross-rack repair bandwidth. Our CP-LRCs advance current research by introducing a cascaded parity design that enhances repair efficiency and system reliability.

\subsection{Wide Stripe Erasure Coding}

Wide-stripe erasure codes are promising strategies for hyperscale data centers due to their ultra-low-cost fault tolerance. In industry, Backblaze \cite{backblaze-erasure-coding}, Vastdata \cite{vastdata-resilience}, and Google \cite{kadekodi23} reported adopting wide stripe designs. In academia, the study \cite{li2017facilitating} used erasure codes with a width of 1024 on hard disks to reduce frequent read retries, while the study \cite{haddock2019high} accelerated wide stripe computation using GPUs. To reduce the cross-rack repair cost, Hu et al. \cite{hu21} utilized Azure's LRC and hierarchical placement, and Yang et al. \cite{yang2022xhr} utilized XOR-Hitchhiker-RS Codes. Kadekodi et al.~\cite{kadekodi23} visited the reliability of wide LRCs and quantitatively analyzed the trade-offs among stripe width, locality, and Mean Time to Data Loss (MTTDL). Several recent studies further explored how to efficiently generate wide stripes from initial narrow stripes \cite{yao2021stripemerge, wu2024optimal}. Our work shows the performance and reliability inefficiencies of existing wide stripe schemes and proposes CP-LRCs that enhance both single- and multi-node repair performance while guaranteeing reliability.

\subsection{System-Level Optimizations for Erasure-Coded Storage}

Prior work has explored a variety of system-level optimizations that are orthogonal to the underlying erasure code structure (including our CP-LRCs). Okapi \cite{athlur25} decouples data striping from redundancy grouping to optimize access performance and reliability jointly. ECPipe \cite{li21} proposed pipelined repair, and subsequent work \cite{wang22} further accelerated repair via parallel disk reconstruction in high-density storage servers. Deterministic data placement has also been leveraged to improve repair performance \cite{li19}. Another line of research focuses on rack-aware repair, aiming to reduce cross-rack bandwidth for single-block repair in Regenerating Codes \cite{hu17, shen2017cross}, RS Codes \cite{hu17, shen2017cross}, and Azure's LRC \cite{wu2020, hu21}. Extending CP-LRCs with rack-awareness remains an interesting direction for future work. Besides, several studies investigate file-level degraded-read optimization for Regenerating Codes \cite{shan2021geometric, niu2026drboost}. Our file-level repair optimization is conceptually similar, but differs in two key aspects: (i) we optimize not only single-node but also multi-node degraded reads, and (ii) we explicitly address repeated-read amplification.
% !TEX root = ../main.tex

\section{Conclusion}

We introduced Cascaded Parity LRCs (CP-LRCs), a new LRC family that makes wide stripes practical by coupling local and global parity blocks through a cascaded structure. CP-LRCs preserve MDS-level fault tolerance while significantly reducing bandwidth for both single-node and multi-node repairs. We developed a general construction framework, efficient repair algorithms, and practical instantiations (CP-Azure and CP-Uniform). Cloud evaluations show substantial repair-time reductions, demonstrating that structured parity cooperation is an effective design principle for scalable, reliable wide stripe storage.

\bibliographystyle{ieeetr}
\bibliography{ref}

\begin{IEEEbiography}[{\includegraphics[width=1in,height=1.25in,clip,keepaspectratio]{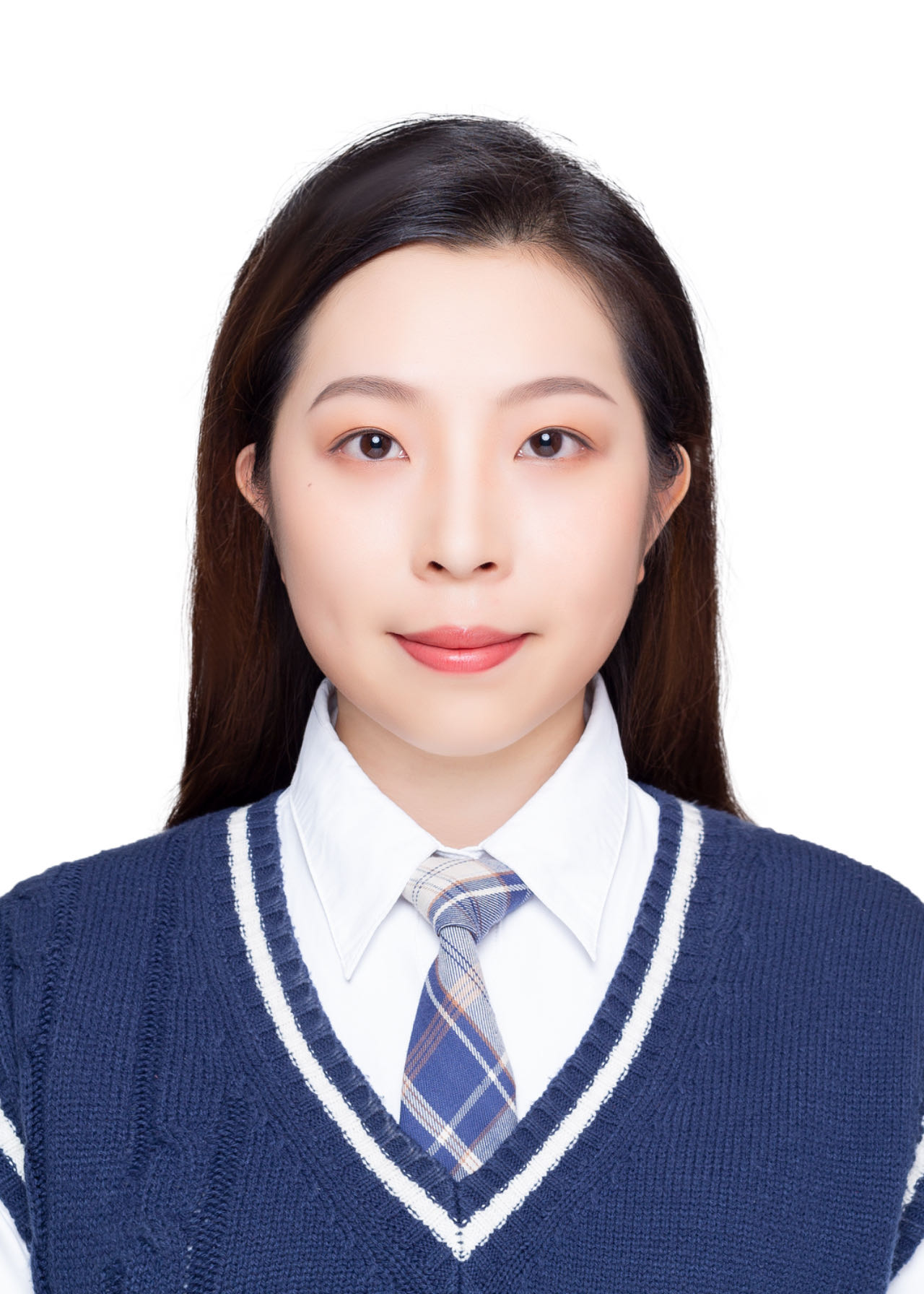}}]
  {Fan Yu} received the B.Eng. degree in Software Engineering from Shandong University in 2023. She is now a Ph.D. student in the School of Cyberspace Information and Technology at Shandong University. Her research interests include distributed systems and erasure coding.
\end{IEEEbiography}

\begin{IEEEbiography}[{\includegraphics[width=1in,height=1.25in,clip,keepaspectratio]{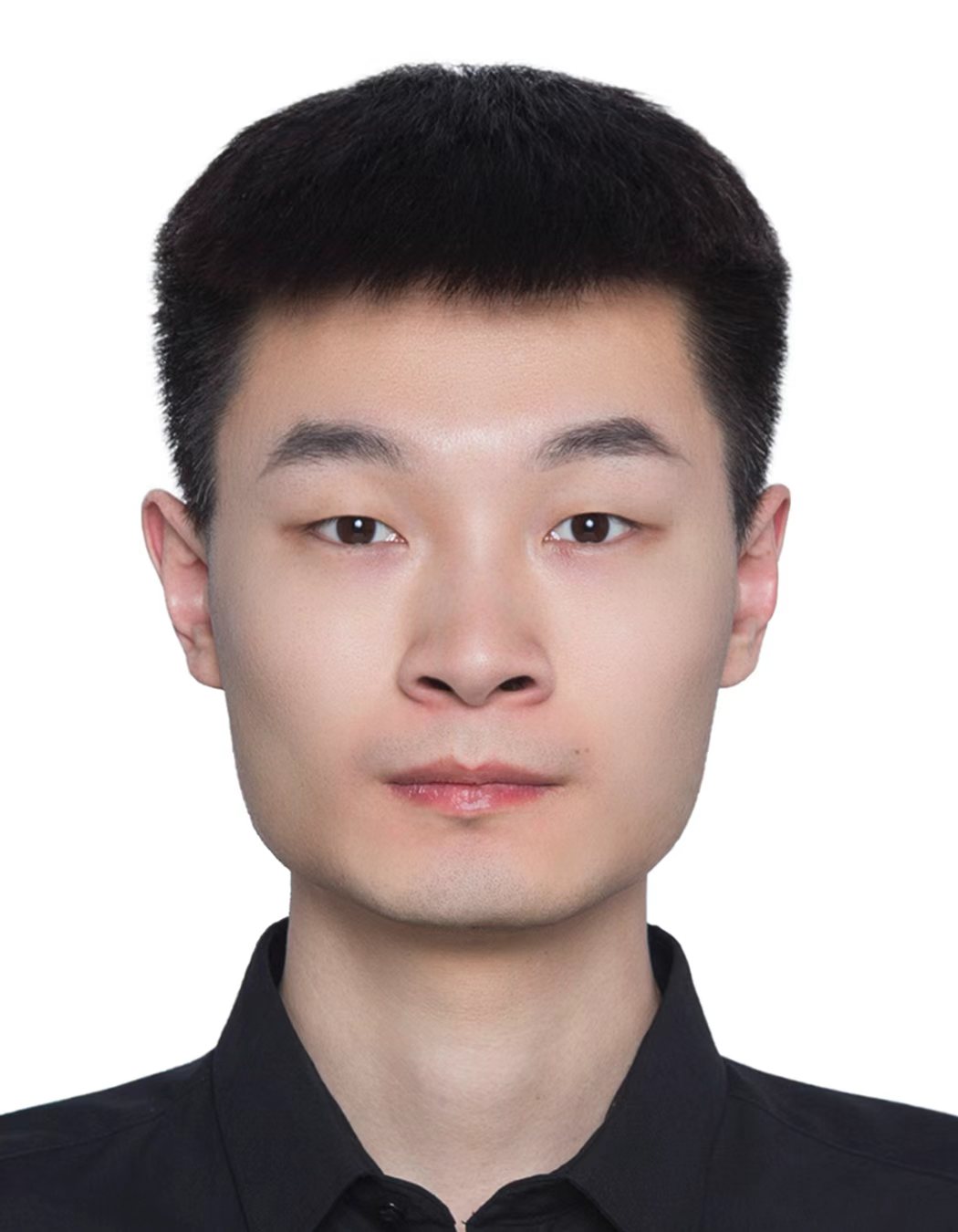}}]
  {Guodong Li} received his B.Eng. degree in computer science and technology from the China University of Petroleum, Qingdao, China, in 2020, and his Ph.D. degree from the School of Cyber Science and Technology, Shandong University, Qingdao, China, in 2025. He is currently a postdoctoral researcher at Shandong University. His research interests focus on channel coding and codes for distributed storage. He received the 2024 Jack Keil Wolf ISIT Student Paper Award.
\end{IEEEbiography}

\begin{IEEEbiography}[{\includegraphics[width=1in,height=1.25in,clip,keepaspectratio]{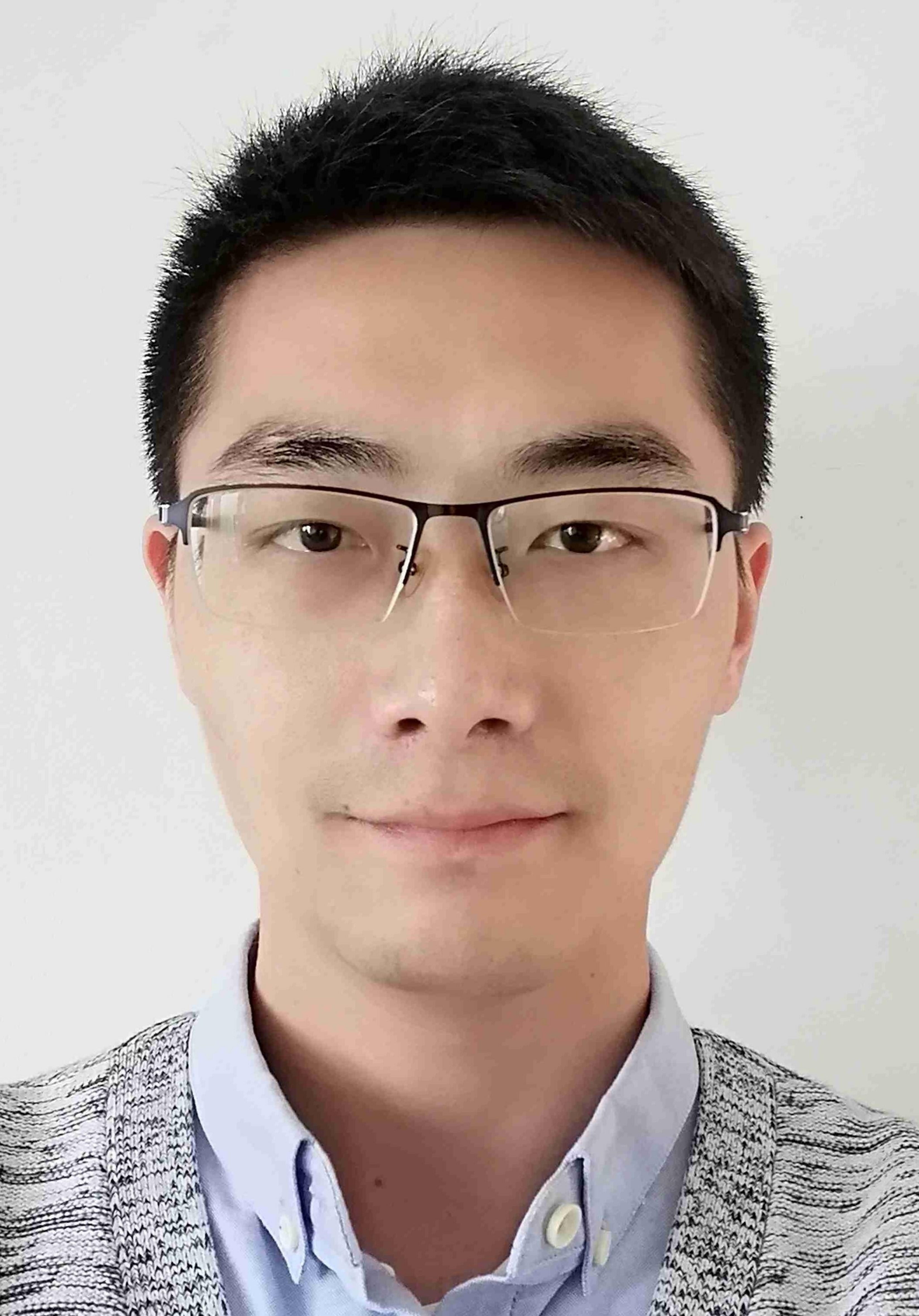}}]
  {Si Wu} received the B.Eng. and Ph.D. degrees in computer science from the University of Science and Technology of China in 2011 and 2016, respectively. He is now a professor in the School of Computer Science and Technology at Shandong University. His research interests include storage reliability and distributed storage.
\end{IEEEbiography}

\begin{IEEEbiography}[{\includegraphics[width=1in,height=1.25in,clip,keepaspectratio]{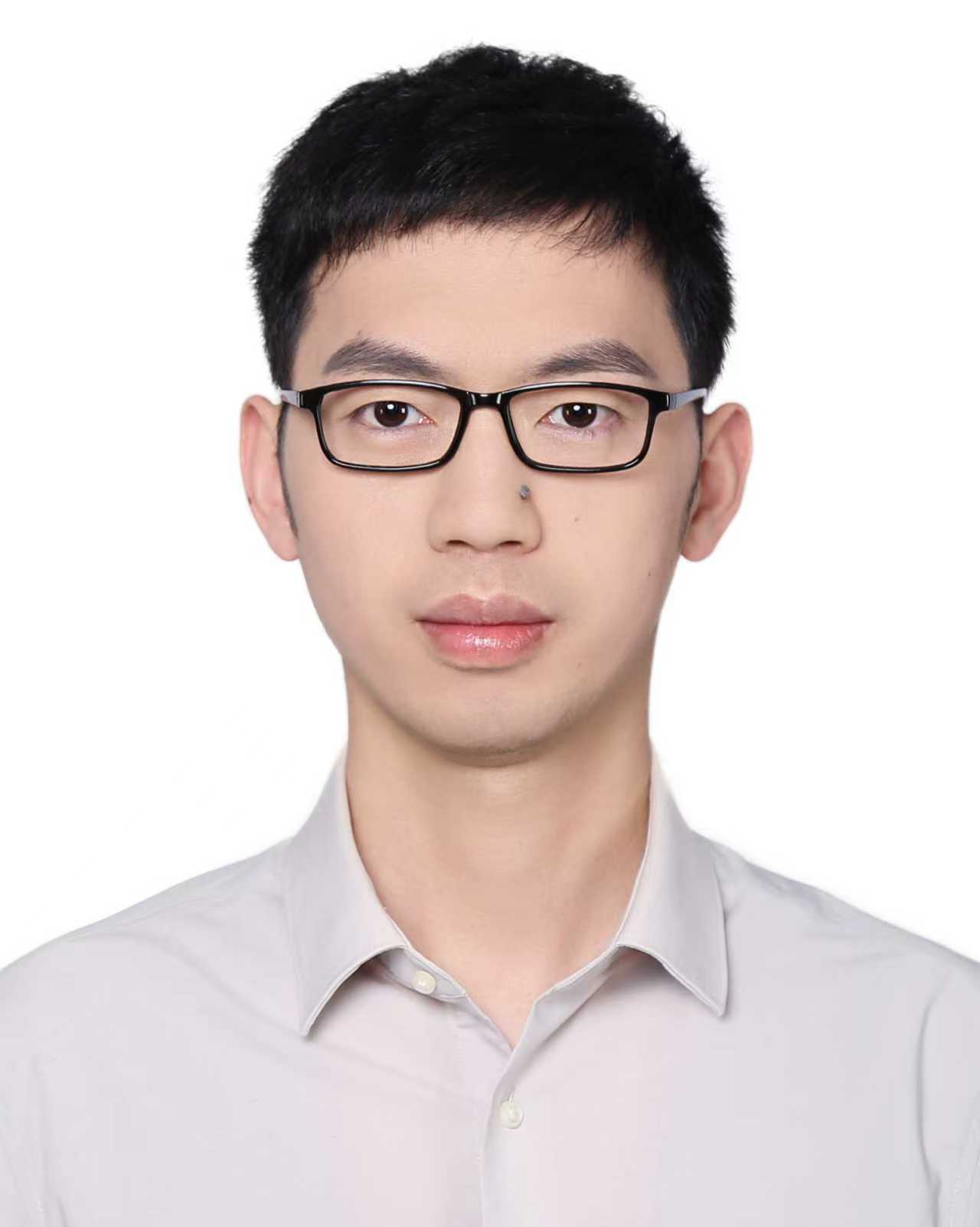}}]
  {Weijun Fang} received the B.S. degree in mathematics and the Ph.D. degree in probability and mathematical statistics from Nankai University, Tianjin, China, in 2013 and 2019, respectively. From 2019 to 2021, he was a Post-Doctoral Researcher with the Tsinghua Shenzhen International Graduate School, Tsinghua University, Shenzhen, China. Since September 2021, he has been a Research Fellow with the School of Cyber Science and Technology, Shandong University, Qingdao, China. He is currently with the Key Laboratory of Cryptologic Technology and Information Security, Ministry of Education, Shandong University. His current research interests include classical coding theory, quantum coding, and coding for distributed storage systems.
\end{IEEEbiography}

\begin{IEEEbiography}[{\includegraphics[width=1in,height=1.25in,clip,keepaspectratio]{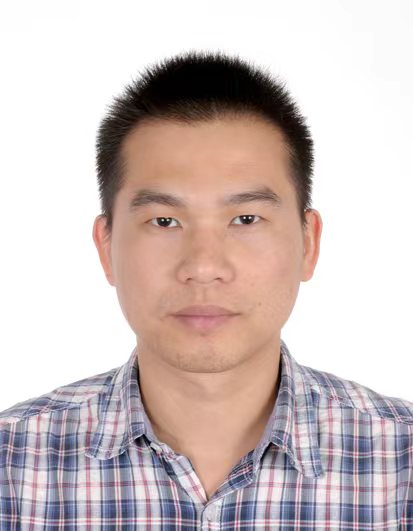}}]
  {Sihuang Hu} received the B.Sc. degree in 2008 from Beihang University, Beijing, China, and the Ph.D. degree in 2014 from Zhejiang University, Hangzhou, China, both in applied mathematics. He is currently a professor at Shandong University, Qingdao, China. Before that, he was a postdoc at RWTH Aachen University, Germany, from 2017 to 2019 and from 2014 to 2015, and a postdoc at Tel Aviv University, Israel, from 2015 to 2017. His research interests include lattices, combinatorics, and coding theory for communication and storage systems. Sihuang is a recipient of the Humboldt Research Fellowship (2017).
\end{IEEEbiography}

\vfill

\clearpage
\setcounter{page}{1}
% !TEX root = ../main.tex

\appendix

\section{Appendix} \label{sec:appendix} Since Cauchy RS codes are the most commonly used MDS codes, we prove that for any $(k,r)$ Cauchy RS codes over GF($2^w$), we can always find the $k+r-1$ combination coefficients $\gamma_1, \dots, \gamma_k$, $\eta_1, \dots, \eta_{r-1}$ satisfy the equation~\eqref{eq:combination}.

\begin{definition}[Cauchy RS code]\label{def:cauchy}
    Let $a_1, a_2, \dots, a_k$ and $b_1, b_2, \dots, b_r$ be $k+r$ distinct elements in GF($2^w$). The corresponding $(k, r)$ Cauchy code can be defined by the following linear combinations from the $k$ data blocks (denoted by $D_1, D_2, \dots, D_k$) to the $r$ parity blocks (denoted by $G_1, G_2, \dots, G_r$):
    \begin{equation}\label{eq:encoding}
        G_j = \alpha_{1,j}D_1 + \alpha_{2, j}D_2 + \cdots + \alpha_{k, j}D_k,
    \end{equation}
    where $\alpha_{i,j} = (b_j-a_i)^{-1}$, $1\le i\le k$, and $1\le j\le r$.
\end{definition}

\begin{theorem}\label{thm:all-combination}
    For a $(k, r)$ Cauchy RS code defined by the $k+r$ distinct elements $a_1, a_2, \dots, a_k$ and $b_1, b_2, \dots, b_r$, there exist $k+r$ nonzero coefficients $\bar\gamma_1, \dots, \bar\gamma_k$, $\bar\eta_1, \dots, \bar\eta_{r}$ such that
    \begin{equation*}
        \bar\gamma_i + \sum_{j = 1}^{r}\bar\eta_j\alpha_{i,j} = 0, \text{~for all~} 1\le i\le k.
    \end{equation*}
\end{theorem}
\begin{proof}
    We give a constructive proof. A specific choice of these nonzero coefficients can be
    \begin{align}
         & \bar\gamma_i = \prod_{z\in[r]}(a_i-b_z)^{-1}, \text{~for all~} 1\le i\le k, \label{eq:bar-gamma} \\
         & \bar\eta_j   = \prod_{\substack{z\in[r]                                                          \\ z\neq j}}(b_j-b_z)^{-1}, \text{~for all~} 1\le j\le r.\label{eq:bar-eta}
    \end{align}
    Firstly, we define the following polynomials over GF$(2^w)$
    \begin{equation}
        p(x) = \prod_{z\in[r]}(x-b_z)
    \end{equation}
    and
    \begin{equation*}
        p_j(x) = \prod_{\substack{z\in[r]\\z\neq j}}(x-b_z) \text{\quad~for all~} 1\le j\le r.
    \end{equation*}
    By performing Lagrange interpolation for constant function $f(x)=1$ at the points $b_1, b_2, \dots, b_r$, we can obtain the following equation:
    \begin{equation*}
        \frac{1}{p(x)} = \sum_{j\in [r]}\frac{1}{x-b_j}\cdot\frac{1}{p_j(b_j)},
    \end{equation*}
    we can derive that, for all $1\le i\le k$,
    \begin{align*}
        \bar\gamma_i + \sum_{j = 1}^{r}\bar\eta_j\alpha_{i,j}=\frac{1}{p(a_i)} + \sum_{j\in [r]}\frac{1}{a_i-b_j}\cdot\frac{1}{p_j(b_j)} = 0.
    \end{align*}
\end{proof}

\begin{corollary}\label{coro:combination}
    For a $(k, r)$ Cauchy RS code defined by the $k+r$ distinct elements $a_1, a_2, \dots, a_k$ and $b_1, b_2, \dots, b_r$, there exist $k+r$ nonzero coefficients $\bar\gamma_1, \dots, \bar\gamma_k$, $\bar\eta_1, \dots, \bar\eta_{r}$ such that
    \begin{equation}\label{eq:all-combination}
        \bar\gamma_1D_1 +\cdots+\bar\gamma_kD_k + \bar\eta_1G_1 + \cdots + \bar\eta_rG_r = \bm 0.
    \end{equation}
\end{corollary}

The proof of Corollary~\ref{coro:combination} is directly from the encoding equations~\eqref{eq:encoding} and Theorem~\ref{thm:all-combination}. Then we can set the needed $k+r-1$ combination coefficients as
\begin{align*}
    \gamma_i = \bar\gamma_i/\bar\eta_r  \text{~and~} \eta_j   = \bar\eta_j/\bar\eta_r,
\end{align*}
where $\bar\gamma_i$ and $\bar\eta_j$ are defined as equation~\eqref{eq:bar-gamma} and \eqref{eq:bar-eta} for $1\leq i\leq k$ and $1\leq j\leq r-1$. According to Corollary~\eqref{coro:combination}, we know that these $k+r-1$ coefficients satisfy the equation~\eqref{eq:combination}, which is required by the CP-Uniform LRCs.

\end{document}